%% file: main.tex
\documentclass[journal,comsoc]{IEEEtran}

\usepackage{amsmath,amsthm,relsize}
\usepackage{amsfonts, amssymb, cuted}
\usepackage{cleveref}
\usepackage{mathtools}
\usepackage{algorithm}
\usepackage[noend]{algpseudocode}
\usepackage{cite}
\usepackage{xcolor}
\usepackage{soul}
\usepackage{graphicx}
\usepackage{tikz}
\usepackage{pgfplotstable}
\usepackage{pgfplots}
\usepackage{nicefrac}
\usepackage{enumitem}
\usepackage{support-caption}
\usepackage{subcaption}
\usepackage[font=footnotesize]{caption}
\usepackage{multirow}
\pgfplotsset{compat=1.15}
\usepackage{relsize}
\usetikzlibrary{patterns}
\newcommand{\rev}[1]{{\color{black} #1}}
\input{MyCommands.tex}

\input{FigureCommands}

\begin{document}

\title{
Low-Complexity High-Performance\\Cyclic Caching for Large MISO Systems
%Cyclic Caching for Low-Complexity High-Performance Communications in Large MISO Systems
\thanks{This work was supported by the Academy of Finland under grants no. 319059 (Coded Collaborative Caching for Wireless Energy Efficiency) and 318927 (6Genesis Flagship), as well as by the ERC project DUALITY (grant agreement no. 725929)}}
\date{September 2020}

\author{
\IEEEauthorblockN{MohammadJavad~Salehi, Emanuele Parrinello, Seyed~Pooya~Shariatpanahi,\\~Petros Elia and Antti~T\"olli
        }
        
\thanks{MohammadJavad Salehi and Antti T\"olli are with the Center for Wireless Communications (CWC), University of Oulu, 90570 Filand, e-mail: \{first\_name.last\_name\}@oulu.fi. Emanuele Parrinello and Petros Elia are with the Communication Systems Department, Eurecom, Sophia Antipolis, 06410 France, e-mail: \{first\_name.last\_name\}@eurecom.fr. Seyed Pooya Shariatpanahi is with School of Electrical and Computer Engineering, College of Engineering, University of Tehran, 1439957131 Iran,
email: p.shariatpanahi@ut.ac.ir}%
\thanks{Part of this work is presented in IEEE International Conference on Communications (ICC), June 2020.}%
%\thanks{Manuscript received September 2020, revised ...}
}

\maketitle

\begin{abstract}
Multi-antenna coded caching is known to combine a global caching gain that is proportional to the cumulative cache size found across the network, with an additional spatial multiplexing gain that stems from using multiple transmitting antennas. However, a closer look reveals two severe bottlenecks; the well-known exponential subpacketization bottleneck that dramatically reduces performance when the communicated file sizes are finite, and the considerable optimization complexity of beamforming multicast messages when the SNR is finite. We here present an entirely novel caching scheme, termed \emph{cyclic multi-antenna coded caching}, whose unique structure allows for the resolution of the above bottlenecks in the crucial regime of many transmit antennas. For this regime, where the multiplexing gain can exceed the coding gain, our new algorithm is the first to achieve the exact one-shot linear optimal DoF with a subpacketization complexity that scales only linearly with the number of users, and the first to benefit from a multicasting structure that allows for exploiting uplink-downlink duality in order to yield optimized beamformers ultra-fast. In the end, our novel solution provides excellent performance for networks with finite SNR, finite file sizes, and many users.
\end{abstract}


\vspace{-4pt}
\section{Introduction}
\input{intro2pe}
\subsection{Prior Work}
\input{LiteratureReview2pe}

\section{System Model}
\label{section:sys_model}
\input{SysModel}
\section{Cyclic Caching for Reduced Subpacketization}
\label{section:ICCPaper}
\input{LinearSubpacks}

%\vspace{-40pt}
\subsection{Further Reduction in Subpacketization}
\label{section:Emanuele_part}
\input{FurtherReduction-Summary}

%\section{Large Antenna Arrays}
%\input{LargerArrays}

\section{Complexity and Performance Analysis}
\label{section:sim_results}
\input{SimResults}

\section{Conclusion and Future Work}
\label{section:conclusion}

\rev{We proposed cyclic caching, a novel low-complexity, high-performance coded caching scheme for large-scale cache-aided MISO networks where the spatial DoF is not smaller than the coded caching gain. For a fixed caching gain, cyclic caching achieves theoretical sum-DoF optimality with a subpacketization requirement that scales linearly with the number of users in the network. Moreover, it delivers all the requested files with a relatively small number of transmissions and does not require the users to decode multiple messages simultaneously. As a result, the scheme can be employed in large networks with many users, at least one order of magnitude larger than other well-known multi-antenna coded caching schemes.
}

%{\color{blue} \\old\\ }
\rev{
%The proposed scheme is centralized because the number of users is known beforehand, and the cache placement algorithm dictates the cache contents of each user. We suspect that applying a very similar approach to a fully decentralized approach where the cache contents of various users are selected randomly would heavily degrade the performance. The reason is that the proposed delivery algorithm depends largely on the existence of similar cache states at different users (cf. Section~\ref{section:Emanuele_part}). However, due to the same reason, we believe that developing a semi-decentralized scheme where the cache states (and not the individual data elements) are selected randomly may result in an acceptable performance even for a moderate number of users. Developing such a caching scheme is part of our current ongoing work. Other topics to be considered include extending the applicability of the proposed scheme to the $t > L$ regime and investigating the effect of non-equal number of interference terms for different users (cf. Section~\ref{section:lin_beamformer}).
Our proposed scheme is centralized, i.e., the server knows the number and identity of the users beforehand and dictates the cached content for each user. A decentralized extension for this work should handle scenarios where such information is unavailable. However, the classic decentralized approach (cf.~\cite{maddah2015decentralized}) where users cache at random a set of bits from the library heavily degrades the performance in the practical finite file size regime considered here. Moreover, our RED scheme depends largely on the existence of similar cached contents at different users, a property that very unlikely is preserved in a fully decentralized approach. Hence, we believe that developing a semi-decentralized scheme, similar to~\cite{jin2019new} where each user picks at random one of the pre-defined cache states, can result in a good performance even in the finite file size regime. Other topics to be considered include extending the applicability of the scheme to the $t > L$ regime and investigating the effect of the non-equal number of interference terms for different users.

%We recall that our proposed scheme is centralized, i.e., the server knows the number and identity of the users beforehand and dictates the cached content for each user. An interesting extension for this work could remove these assumptions to handle scenarios where such information about the users is unavailable during the placement phase. However, the classical decentralized approach (cf. \cite{maddah2015decentralized}), where each user caches at random a set of bits from the library, is known to heavily degrade the performance in the relevant and practical finite file size regime for which our centralized scheme here is developed. Moreover, our RED scheme discussed in Section~\ref{section:Emanuele_part} depends largely on the existence of similar cached contents at different users, a property that very unlikely can be maintained if we use a fully decentralized approach. For these reasons, we believe that developing a semi-decentralized scheme, in the same line as the approach in~\cite{jin2019new} where each user picks at random one of the pre-defined cache states, can result in a good performance even for a moderate number of users and limited file sizes. Developing such a semi-centralized coded caching scheme is part of our ongoing work.  Other topics to be considered include extending the applicability of the proposed scheme to the $t > L$ regime and investigating the effect of the non-equal number of interference terms for different users (cf. Section~\ref{section:lin_beamformer}).
}
\appendix
\subsection{More Detailed Analysis of the Delivery Phase}
\input{LinearS-ValidityComplexity}
% \vspace{-10pt}
\subsection{Reducing Subpacketization by a Factor of $\phi_{K,t,\alpha}^2$}
\input{FurtherReduction}

\bibliographystyle{IEEEtran}
\bibliography{references}

\end{document}

%% file: MyCommands.tex
\newtheorem{thm}{Theorem}%[section]
\newtheorem{exmp}{Example}%[section]
\theoremstyle{definition}

\newtheorem{rem}{Remark}

\newcommand{\CA}[0]{{\mathcal{A}}}
\newcommand{\CB}[0]{{\mathcal{B}}}

\newcommand{\CD}[0]{{\mathcal{D}}}

\newcommand{\CF}[0]{{\mathcal{F}}}

\newcommand{\CR}[0]{{\mathcal{R}}}

\newcommand{\CX}[0]{{\mathcal{X}}}

\newcommand{\CZ}[0]{{\mathcal{Z}}}

\newcommand{\Ba}[0]{{\mathbf{a}}}

\newcommand{\Bc}[0]{{\mathbf{c}}}
\newcommand{\Bd}[0]{{\mathbf{d}}}
\newcommand{\Be}[0]{{\mathbf{e}}}

\newcommand{\Bh}[0]{{\mathbf{h}}}

\newcommand{\Bk}[0]{{\mathbf{k}}}

\newcommand{\Bp}[0]{{\mathbf{p}}}

\newcommand{\Br}[0]{{\mathbf{r}}}
\newcommand{\Bs}[0]{{\mathbf{s}}}

\newcommand{\Bu}[0]{{\mathbf{u}}}
\newcommand{\Bv}[0]{{\mathbf{v}}}
\newcommand{\Bw}[0]{{\mathbf{w}}}
\newcommand{\Bx}[0]{{\mathbf{x}}}

\newcommand{\BD}[0]{{\mathbf{D}}}

\newcommand{\BG}[0]{{\mathbf{G}}}

\newcommand{\BI}[0]{{\mathbf{I}}}

\newcommand{\BV}[0]{{\mathbf{V}}}

%% file: FigureCommands.tex
\newcommand{\MyRect}[2]{(#2-1+0.1,6-#1+0.1) rectangle (#2-0.1,6-#1+1-0.1)}
\newcommand{\FillBlack}[2]{\filldraw[gray!50]\MyRect{#1}{#2}}
\newcommand{\FillGray}[2]{\filldraw[black!70]\MyRect{#1}{#2}}

\newcommand{\MyRectNf}[2]{(#2-1+0.1,4-#1+0.1) rectangle (#2-0.1,4-#1+1-0.1)}
\newcommand{\FillBlackNf}[2]{\filldraw[gray!50]\MyRectNf{#1}{#2}}
\newcommand{\FillGrayNf}[2]{\filldraw[black!70]\MyRectNf{#1}{#2}}

\usepackage{ellipsis}
\usetikzlibrary{calc,quotes}
\usetikzlibrary{decorations.pathreplacing,decorations.markings,shapes.geometric}
\tikzset{naming/.style={align=center,font=\small}}
\tikzset{rotatenaming/.style={align=center,font=\tiny}}
\tikzset{antenna/.style={insert path={-- coordinate (ant#1) ++(0,0.25) -- +(135:0.25) + (0,0) -- +(45:0.25)}}}
\tikzset{station/.style={naming,draw,shape=dart,shape border rotate=90, minimum width=10mm, minimum height=10mm,outer sep=0pt,inner sep=3pt}}
\tikzset{mobile/.style={naming,draw,shape=rectangle,minimum width=12mm,minimum height=6mm, outer sep=0pt,inner sep=3pt}}
\tikzset{radiation/.style={{decorate,decoration={expanding waves,angle=90,segment length=4pt}}}}
\tikzset{user terminal/.style={naming,draw,shape=rectangle,minimum width=3mm,minimum height=5mm, outer sep=0pt,inner sep=3pt}}
\tikzset{user cache/.style={rotatenaming,draw,shape=cylinder,minimum width=7mm,minimum height=5mm, outer sep=0pt,inner sep=3pt, shape border rotate=90}}

\tikzset{zigzag/.style = {% added for solution
    to path={ -- ($(\tikztostart)!.55!-9:(\tikztotarget)$) --
                 ($(\tikztostart)!.45!+9:(\tikztotarget)$) -- (\tikztotarget)
             \tikztonodes},sharp corners}
        }

\newcommand{\MBase}[1]{%
\begin{tikzpicture}
\node[station] (base) {#1};

\draw[line join=bevel] (base.100) -- (base.80) -- (base.110) -- (base.70) -- (base.north west) -- (base.north east);
\draw[line join=bevel] (base.100) -- (base.70) (base.110) -- (base.north east);

\draw[line cap=rect] ([xshift=-.1768cm,yshift=.6pt]base.north -| base.right tail) [antenna=1];
\draw[line cap=rect] ([yshift=.6pt]ant1 |- base.north) -- node[above,shape=rectangle,inner ysep=+.3333em]{\dots} ([xshift=.1768cm,yshift=.6pt]base.north -| base.left tail) [antenna=2];

\end{tikzpicture}
}

\newcommand{\UT}[1]{%
\begin{tikzpicture}[every node/.append style={rectangle,minimum width=0pt}]
\node[user terminal] (box) {};
\node[user cache] at (1.0,0) (cache) {};

\draw []  node at (1.0,0.01) {\tiny{$\CZ(#1)$}};
\draw [] (box.east) to (cache.west);

\draw (box.north) [antenna=1];
\end{tikzpicture}
}

%% file: intro2pe.tex
Wireless communication networks are under mounting pressure to support the exponentially increasing volumes of multimedia content, as well as to support the imminent emergence of applications such as wireless immersive viewing and reliable autonomous driving~\cite{cisco2018cisco}. For the efficient delivery of such multimedia content, the work of  Maddah-Ali and Niesen~\cite{maddah2014fundamental} proposed the idea of \emph{coded caching} as a means of increasing the data rates by exploiting cache content across the network. This approach considered a single-stream (single-antenna) downlink network of cache-enabled receiving users who can pre-fetch data into their cache memories, in a way that --- during the subsequent content delivery --- the achievable rates can be boosted by a multiplicative factor that is proportional to the cumulative cache size across the entire network. The key to achieving this speedup was the ability of coded caching to successfully multicast individual messages to many users at a time, such that each user can use its cache content to remove unwanted messages from the received signal.  In this context, a system that can multicast to $t+1$ users at a time is said to enjoy a (cache-aided) Degrees-of-Freedom (DoF) performance $t+1$, which also matches the aforementioned multiplicative speedup factor in the high-SNR regime. This quantity $t$ --- which is often referred to as the (additive) \emph{coded caching gain} --- depends on the number of receiving users and the size of their cache, and it effectively describes the redundancy with which data can be cached across the network. %(or equivalently, the required bandwidth can be decreased ) 

Motivated by the unavoidable dominance of multi-antenna paradigms in wireless communications~\cite{rajatheva2020white}, Shariatpanahi et al.~\cite{shariatpanahi2016multi,shariatpanahi2018physical} explored the cache-aided multi-antenna setting, for which it revealed that this same caching gain could, in fact, be maintained in its fullest. In a basic downlink scenario with $L$ such transmit antennas, the work in~\cite{shariatpanahi2018physical} developed a method that achieved a DoF of $t+L$, which was shown to be optimal under basic assumptions of uncoded cache placement and one-shot data delivery (cf.~\cite{lampiris2018resolving}).  
%to be optimal under the assumption of uncoded cache placement and one-shot linear schemes. %This ability of coded caching to additively combine caching and spatial-multiplexing gains, strongly advocates for fusing the paradigms of caching and multi-antenna interference management. 

However, despite the original theoretical promises for large caching gains, in reality, coded caching can suffer from severe bottlenecks that dramatically limit these gains. Undoubtedly the most damaging of these is the well-known \emph{subpacketization bottleneck}, which stems from the fact that the aforementioned caching gains require (cf.\cite{shanmugam2016finite},\cite{yan2017placement}) that each file is split into a number of subfiles that scales exponentially with the number of receiving users. This requirement is exacerbated in multi-antenna coded caching approaches (cf.~\cite{shariatpanahi2018physical}), where the subpacketization (and thus the file-size requirements) can far exceed those of the original scheme in~\cite{maddah2014fundamental}. In essence, subpacketization requirements rendered coded caching hard to implement in most moderate- or large-sized networks. % --- persisted in the majority of multi-antenna coded caching approaches of~\cite{shariatpanahi2018physical}, where the subpacketization (and thus the file-size requirements) can far exceed those of the original scheme in~\cite{maddah2014fundamental}.

While though the subpacketization requirements for achieving full caching gains in the single antenna setting are indeed fundamental and generally unavoidable (cf.\cite{shanmugam2016finite},\cite{yan2017placement}), a recent new approach in~\cite{lampiris2018adding} has shown that these limitations are not fundamental in the multi-antenna setting. As we now know from~\cite{lampiris2018adding}, activating multiple transmit antennas can effectively decompose the cache-aided network in a manner that dramatically alleviates the subpacketization bottleneck, thus strongly boosting the true (subpacketization constrained) performance. 
%By designing a scheme that effectively decomposes the network,~\cite{lampiris2018adding} showed that if in the single antenna setting, coded caching system can provide --- under subpacketization constraints --- an additive coded caching gain of $\overline{t}$, then simply by activating $L$ transmit antennas, this caching gain is \emph{multiplicatively} boosted up to $L\overline{t}$.  
This performance boost, as well as the ability to achieve the theoretical promises of multi-antenna coded caching with exceedingly small subpacketization complexity, offers a powerful motivation for meaningfully combining coded caching with multi-antenna communications. 

It is the case though that, to date, this decomposition principle as it was developed in~\cite{lampiris2018adding}, applies with full optimality only in networks with essentially a modest number of transmit antennas, and in particular, only in networks where the spatial multiplexing gain (which can go up to $L$) does not exceed the coded caching gain $t$.  In the last two years, extending this optimality to networks with `larger' antenna arrays has been a known open problem, which is indeed motivated by the upcoming prevalence of very large antenna arrays. 

We present a new multi-antenna coded caching structure that resolves this theoretical problem; it achieves the exact optimal one-shot linear DoF $L+t$ even when $L \ge t$ and does so with ultra-low subpacketization. Perhaps more importantly, our scheme, which employs a novel cyclic structure, additionally allows for an ultra-efficient implementation of optimized precoders, thus yielding excellent performance in the low- and mid-SNR range for even large networks.

%% file: LiteratureReview2pe.tex
\subsubsection{Single- and multi-antenna coded caching}
The original coded caching scenario in~\cite{maddah2014fundamental} considers a single-antenna transmitter that communicates to $K$ cache-aided receiving users via a single-stream, normalized-capacity, symmetric broadcast channel. The transmitter hosts a library of $ N $ equal-sized files, and each user has a cache memory of size equal to the size of $ M $ files. The setting considers two distinct phases: the cache placement phase, which occurs during the off-peak hours, and the content delivery phase, which occurs during the peak hours. The first phase allows for the users' cache memories to be filled with library content, and the second phase employs basic cache-aided multicasting in order to serve $t+1$ users at a time, thus yielding a DoF of $t+1$, where $t\coloneqq \frac{KM}{N}$. Decoding is achieved by having each user exploit their locally cached content in order to remove the $t$ undesired messages from the received signals. 
% From here in revision 1
This scheme was proven to be exactly optimal under the constraint of uncoded cache placement~\cite{yu2017exact,wan2016optimality} and optimal within a gap of two in the general case~\cite{yu2018characterizing}. Soon after~\cite{maddah2014fundamental}, the work in~\cite{maddah2015decentralized} (see also~\cite{girgis2017decentralized,xu2018fundamental}) proposed a decentralized version of the scheme, which, in order to provide good gains, required the file sizes to be very large. In contrast, the work in~\cite{jin2019new} proposed a single-antenna decentralized scheme for finite-size files, but with significantly worse performance than~\cite{maddah2014fundamental,maddah2015decentralized}.
Apart from the decentralized setting, the use of coded multicasting was extended to other scenarios such as hierarchical, device-to-device (D2D), and fog radio access networks~\cite{karamchandani2016hierarchical,ji2015fundamental,roig2018storage}. One of the most important scenarios is the multi-antenna/multi-transmitter case, for which the scheme in the aforementioned work in~\cite{shariatpanahi2016multi} achieves a sum-DoF of $t+L$, which was later proved in~\cite{naderializadeh2017fundamental} to be optimal within a factor of two among all linear one-shot schemes (this gap was then tightened in~\cite{lampiris2018resolving}). The cache-aided interference channel studied in~\cite{naderializadeh2017fundamental} was later extended to cellular networks~\cite{naderializadeh2019cache} and heterogeneous parallel channels with\textbackslash without CSIT~\cite{piovano2019centralized}. Unlike~\cite{naderializadeh2017fundamental}, the work in~\cite{hachem2018degrees} characterized the global sum-DoF of the cache-aided interference channel (without any restriction on the type of schemes) by proposing a method based on interference alignment techniques that achieves an approximately optimal DoF.. In the specific case where each transmitter can cache the entire library (coinciding with the multi-antenna coded caching problem), the scheme in~\cite{hachem2018degrees} achieves a larger sum-DoF than the optimal linear single-shot sum-DoF. However, the complexity of interference alignment makes the scheme very far from practical, which contrasts with our paper that aims to provide a practical scheme with low complexity.  In this regard, we consider only one-shot linear schemes throughout this paper, and by optimal sum-DoF, we mean optimal achievable sum-DoF among all such schemes.

\subsubsection{Multi-antenna coded caching in the finite-SNR regime}
The implementation of coded caching techniques in multi-antenna wireless networks initially emphasized the (high-SNR) DoF setting and thus focused on using basic zero-forcing (ZF) precoders~\cite{shariatpanahi2017multi,shariatpanahi2018physical}.  
Subsequently, the emphasis was shifted on the lower SNR regimes, with the work in~\cite{tolli2018multicast} replacing ZF precoders with optimized beamformers that introduced the capability of controlling, rather than completely nulling out, the inter-stream interference. The subsequent work in~\cite{tolli2017multi} emphasized both the performance and complexity of beamformer designs, introducing novel optimized precoders that properly controlled the multiplexing gain and the size of the corresponding multiple-access channel (MAC) that is experienced during decoding. In particular, it was shown how operating at a multiplexing gain that is smaller than $L$ can reduce beamformer design complexity as well as yield higher beamforming gains, which are crucial in the low-SNR regime. Controlling the spatial multiplexing gain is also considered in~\cite{lampiris2019bridgingb}, where numerical simulations are used to find the best multiplexing gain for various network parameters, including the coded caching gain. Taking another point of view, the work in~\cite{zhao2020multi} proposed a new scheme that limits the number of messages received by a user in each time slot, in order to reduce complexity while maintaining satisfactory rates. Similar works include~\cite{shariatpanahi2018multi,lampiris2019fundamental,bergel2018cache}. While numerical evaluations suggest that the aforementioned schemes can perform well, this performance is limited to very small network scenarios, mainly due to their massive subpacketization requirements.  

\subsubsection{Subpacketization bottleneck}
To date, in the single-antenna setting, any high-performance coded caching scheme requires a subpacketization that grows --- for fixed $\frac{M}{N}$ --- exponentially or near-exponentially with $K$.
% From here in revision 1
The scheme in~\cite{maddah2014fundamental} requires subpacketization $\binom{K}{t}$, and as we know from~\cite{shanmugam2016finite}, decentralized schemes (cf.~\cite{maddah2015decentralized}) also require exponential subpacketization in order to achieve linear caching gains. Along similar lines, we know from~\cite{yan2018placement} that under basic assumptions, there exists no single-antenna coded caching scheme that enjoys both linear caching gains and linear subpacketization. Nevertheless, several works in the literature have tried to reduce the required subpacketization in the shared-link coded caching problem.
Notably, in~\cite{yan2017placement}, Placement Delivery Array (PDA) is proposed as a systematic approach for reducing subpacketization in centralized schemes. For a very large number of users $K$, in~\cite{shangguan2018centralized}, a hypergraph-based scheme is introduced to trade-off subpacketization with performance, and in~\cite{shanmugam2017coded}, it is shown that linearly scaling gains (with $K$) are achievable with a subpacketization that also scales almost linearly. 
%Specifically, the work in~\cite{yan2017placement} proposed Placement Delivery Array (PDA) as a systematic approach to reducing subpacketization in centralized schemes. However, the achievable performance of these PDA-driven schemes is far below the optimal one achieved by~\cite{maddah2014fundamental}. 
%In~\cite{shangguan2018centralized}, the authors used hypergraphs to propose a scheme that nicely trades-off subpacketization with performance; requiring a very large number of users. The work in~\cite{shanmugam2017coded} shows that in the unpractical regime where the number of users $K$ is extremely large, it is possible to achieve gains scaling linearly with $K$ with a subpacketization that also scales almost-linearly with $K$. 
A more recent effort was made in~\cite{chittoor2020subexponential}, where line graphs and projective geometry are used to achieve an asymptotically (as $K\rightarrow \infty$) constant performance with a sub-exponential subpacketization as long as the users' caches are sufficiently large. In the same line of research, other interesting works are~\cite{tang2018coded,michel2020placement,mahesh2020coded}. However, despite the large literature on the topic, no scheme is known to achieve good gains with a relatively small subpacketization and practical values for system parameters.

In the context of multi-antenna coded caching, the situation is different. While the original scheme in~\cite{shariatpanahi2018physical} required an astronomical subpacketization of $\binom{K}{t} \binom{K-t-1}{L-1}$, the recent work in~\cite{lampiris2018adding} showed that if $\frac{K}{L}$ and $\frac{t}{L}$ are both integers, the optimal DoF $t+L$ is achievable with a subpacketization of~$\binom{K/L}{t/L}$, which is dramatically less than the subpacketization in~\cite{maddah2014fundamental,shariatpanahi2018physical}. This directly means that under fixed subpacketization constraints (fixed file size), adding multiple antennas can multiplicatively boost the real (subpacketization-constrained) DoF by a factor of $L$.
The main idea behind the work in~\cite{lampiris2018adding} is to employ basic user-grouping techniques to endow groups of users with the same cache content and then apply a specific precoding approach that decomposes the network of users into effectively parallel coded caching problems. While, as we know from~\cite{parrinello2019fundamental}, for single-antenna setups, having shared caches between the users causes an inevitable DoF loss, the work in~\cite{lampiris2018adding} has proven that multi-antenna shared-cache setups need not suffer from DoF losses. Of course, this is valid under the assumption that $\frac{K}{L}$ and $\frac{t}{L}$ are integers.\footnote{The scheme of~\cite{lampiris2018adding} suffers DoF losses (and also increased subpacketization) if either $\frac{K}{L}$ or $\frac{t}{L}$ is non-integer. The DoF is reduced
by a multiplicative factor that can reach 2 when $L > t$, and can reach $\frac{3}{2}$ when $L < t$.}

Another interesting work can be found in~\cite{mohajer2020miso}, which proposes a DoF-optimal scheme that yields a reduction in transmission and decoding complexity compared to the optimized beamformer scheme of~\cite{tolli2017multi}, albeit with a small reduction in performance compared to \cite{tolli2017multi}, and also with an exponential subpacketization $\binom{K}{t}$. In another line of work,~\cite{lampiris2019bridging} provides a novel algorithm that reduces the channel state information (CSI) requirements, and does so with subpacketization $L_c \binom{K_c}{t}$, where $L_c \coloneqq \frac{L+t}{t+1}$ and $K_c \coloneqq \frac{K}{L_c}$. Finally,~\cite{salehi2019subpacketization,salehi2019beamformer} explore, under the assumption of $K = t+L$, how subpacketization can be traded-off with performance. To date, existing multi-antenna schemes either exhibit subpacketization requirements that are exponential in $K$, or do not experience DoF optimality in scenarios where $L>t$.

%However, the proposed scheme requires both $L_c$ and $K_c$ to be integers, making it applicable to a narrow set of network parameters. Moreover, it results in a DoF loss by a factor of $(1-\frac{t}{K})$.
%\vspace{-5pt}
\subsection{Our Contribution}
Motivated by the ever-increasing sizes of antenna arrays, we proceed to bridge the aforementioned gap and provide a very low-subpacketization coded caching algorithm that achieves the optimal one-shot linear DoF of $L+t$ even in networks with `larger' antenna arrays such that $L\geq t$. 
Subsequently, knowing well that in the low-to-moderate SNR regimes beamforming gains can be as important as multiplexing gains, we proceed to consider the more general scenario where the multiplexing gain $\alpha \le L$ is traded off with an ability to beamform in a manner that compensates the well-known effects of the worst-user channel condition. In our case, the multiplexing gain $ \alpha \ge t$ is treated as a design parameter calibrated not only for yielding excellent finite-SNR performance but also for fine-tuning subpacketization and complexity requirements. 

The proposed multi-antenna \emph{cyclic-caching scheme} enjoys a novel structure that attains the chosen sum-DoF of $t+\alpha$, as well as the holy trinity of high beamforming gains, reduced subpacketization, and reduced beamforming complexity. In particular, the scheme requires an ultra-low subpacketization  $\frac{K(t+\alpha)}{\phi_{K,t,\alpha}^2}$, where $\phi_{K,t,\alpha} = gcd(K,t,\alpha)$; this is currently the smallest known subpacketization out of all one-shot linear schemes with optimal sum-DoF. %and, as we will see later on, it can be fine-tuned by calibrating $\alpha$ and even $K$ (we will explain this latter aspect, later on).
Most importantly, our cyclic caching method also eliminates the requirement of multicasting, thus interestingly enabling optimized beamformers to be designed with massively reduced computational complexity. These optimized beamformers allow for an interplay between cache-based cancellation, nulling out, and controlling the interference, resulting in a very considerable performance boost in low- and mid-SNR communications. In the end, our work provides a method for a high-performance hyper-efficient use of coded caching in large multi-antenna networks and for any SNR value. 

%\vspace{-5pt}
\subsection{Structure and Notation}
We use $[K]$ to denote the set $\{1,2,...,K\}$ and $[i:j]$ to represent the vector $[i \;\, i+\!1 \;\, ... \;\, j]$. Boldface upper- and lower-case letters denote matrices and vectors, respectively. $\BV[i,j]$ refers to the element at the $i$-th row and $j$-th column of matrix $\BV$, and $\Bw[i]$ represents the $i$-th element in vector $\Bw$.
Moreover, $\Bw=[\Bu;\Bv]$ refers to a vector $\Bw$ that is formed by concatenating vectors $\Bu$ and $\Bv$. Sets are denoted by calligraphic letters. For two sets $\CA$ and $\CB$, $\CA \backslash \CB$ refers to the set of elements in $\CA$ that are not in $\CB$, and $|\CA|$ is the number of elements in $\CA$.

The rest of the paper is organized as follows. Section~\ref{section:sys_model} presents the system model, while Section~\ref{section:ICCPaper} presents the new cyclic caching scheme. In this section, we show that for any $\alpha \ge t$, a sum-DoF of $t+\alpha$ is possible, first with subpacketization $K(t+\alpha)$, and then with a subpacketization that is further reduced by a factor of $\phi_{K,t,\alpha}^2$. Section~\ref{section:sim_results} presents complexity analysis and performance simulations, while Section~\ref{section:conclusion} concludes the paper.

%% file: SysModel.tex
\subsection{Network Setup}
We consider a multiple-input, single-output (MISO) broadcast setting, where a single server, equipped with $ L $ transmit antennas, communicates with $ K $ single-antenna receiving users over a shared wireless link. An illustration of the considered communication setup for a small network of $K=3$ users is provided in Figure~\ref{fig:sys_architecture}.
The server has access to a library $\CF$ of $N \ge K$ files, where each file $W \in \CF$ has a size of $f$ bits. We assume that every user has a cache memory of size $Mf$ bits. As mentioned earlier, we use $t \coloneqq K\frac{Mf}{Nf}$ to denote the total cache size in the network normalized by the size of the library. In essence, $t$ --- which we will assume to be an integer --- indicates the redundancy with which the library can be stored across the network.

\begin{figure}[t]
    \centering
    \resizebox{0.70\columnwidth}{!}{%
    \begin{tikzpicture}
    \matrix[column sep=3.5cm,row sep=0cm,ampersand replacement=\&]
    {
        \& \node(UT1)[label=below:\tiny{User 1}]{\UT{1}}; \\
        \node(MBS)[label=below:\tiny{Server}]{\MBase{}}; \&    \node(UT2)[label=below:\tiny{User 2}]{\UT{2}}; \\
        \& \node(UT3)[label=below:\tiny{User 3}]{\UT{3}}; \\
    };
    \draw[shorten >=1mm,->] (MBS.east) to [zigzag,"\tiny{$\Bh_1$}"] (UT1.west);
    \draw[shorten >=1mm,->] (MBS.east) to [zigzag,"\tiny{$\Bh_2$}"] (UT2.west);
    \draw[shorten >=1mm,->] (MBS.east) to [zigzag,"\tiny{$\Bh_3$}"] (UT3.west);
    \end{tikzpicture}
    }
    \caption{Illustration of the communication setup for a network of $K=3$ users}
    \label{fig:sys_architecture}
\end{figure}
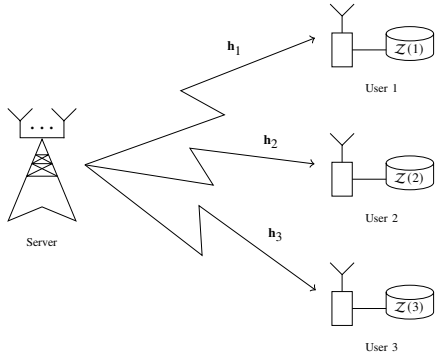
%\vspace{-5pt}

During the cache placement phase, the placement algorithm operates without any prior knowledge of future requests. We use $\CZ(k)$ to denote the cache contents of user $k \in [K]$ after the placement phase is completed.
%The system operates in two distinct phases: placement and delivery. During the placement phase, which is assumed to occur during the low network traffic time, the users' cache memories are filled with data from the files in $ \CF $. This operation is done following a cache placement algorithm, which operates without any prior knowledge of future requests. 
%In this work, we focus on symmetric uncoded cache placement strategies such that each file $W\in \CF$ is split into the same number equal-sized parts, where the number of these parts is called subpacketization. We use $\CZ(k)$ to denote the cache contents of user $k \in [K]$ after the placement phase is completed.
At the beginning of the delivery phase, each user $k \in [K]$ reveals its requested file $W(k) \in \CF$ to the server. After receiving the demand set $\CD = \{W(k) \mid k\in [K] \}$, the server follows the delivery algorithm to transmit the requested subpackets to the users. This will involve the transmission of some $I$ transmission vectors $\{ \Bx_i \} \in \mathbb{C}^L$, $i \in [I]$, where $I$ is given by the delivery algorithm. These transmission vectors are transmitted in consecutive time intervals or separate frequency bins,\footnote{For comparison with other schemes, we will generally assume that transmissions here are made in consecutive time intervals.} using the array of $L$ antennas.
After $\Bx_i$ is transmitted, user $k$ receives
%\begin{equation}
%\label{eq:reception_model}
    $y_i(k) = \Bh_k^H \Bx_i + w_i (k)$,
%\end{equation}
where $\Bh_k \in \mathbb{C}^{L}$ denotes the channel vector and $w_i(k) \sim \mathbb{C}\mathcal{N}(0,N_0)$ denotes the observed noise at user $k$. Furthermore, we consider a slow-fading model in which the channel vectors remain constant during each time interval $i$, and we assume that full channel state information (CSI) is available at the server.\footnote{Improving CSI accuracy is a well-studied topic in the literature. In general, the designs are specific to the underlying uncertainty model, i.e., whether CSI error is bounded~\cite{vorobyov2003robust} or unbounded~\cite{Komulainen2013CoordinatedMode}. If the error is bounded, robust (worst-case) beamforming solutions can be computed that guarantee the achievability of the max-min SINR for all possible realizations of channel uncertainty. For unbounded error, statistically robust solutions can be provided if the estimation noise is assumed to follow a known distribution. In such a scenario, fixed-point iterations are used to provide a robust solution by adding the noise contribution from CSI uncertainty to the thermal noise.}

Let $\CX_i \subseteq [K]$ denote the set of users targeted by $\Bx_i$, and let $T_i$ denote the duration of time interval $i$ required so that every user in $\CX_i$ decodes its intended data from $\Bx_i$. If we consider $L_i$ to be the length of the codeword transmitted at time slot $i$, and if we define $R_i$ to be the multicast rate at which the server transmits a common message to all users in $\CX_i$, then $T_i$ is simply the ratio between $L_i$ and $R_i$.
We will use the metric of the \textit{symmetric rate}, which describes the total number of bits per second with which each user is served. Notably, we will consider the worst-case metric, corresponding to the symmetric rate at which the system can serve all users in the network irrespective of the demand set $\CD$. Given that the delivery phase has an overall duration of $\sum_{i=1}^I T_i$, and given that there are $K$ users, the symmetric rate can be computed as
\begin{equation}
\label{eq:sym_rate}
    R_{sym}=\frac{Kf}{\sum_{i=1}^I T_i}=\frac{Kf}{\sum_{i=1}^I \frac{L_i}{R_i}}.
\end{equation}
Our aim is to design a placement and delivery scheme that maximizes $R_{sym}$. 

\subsection{Building the Transmission Vectors}
We use linear precoding to build the transmission vectors. A generic transmission vector $\Bx_i$ is built as
\begin{equation}
\label{eq:general_transmission_vector_model}
    \Bx_i=\sum_{k\in \CX_i} \Bw_{i}(k) X_{i}(k) \; ,
\end{equation}
where $X_{i}(k)$ is the data codeword transmitted to user $k$, and $\Bw_{i}(k) \in \mathbb{C}^L$ is the beamforming vector used for $X_i(k)$. The beamforming vectors $\Bw_{i}(k)$ are here designed to maximize the worst-user rate, or equivalently, the worst user's SINR (signal to interference and noise ratio), as $\Bx_i$ is transmitted. Thus, given the transmission model in \eqref{eq:general_transmission_vector_model}, the multicast rate at time interval $i$ is calculated as
\begin{equation}%\label{eq:Ri}
  R_i=\log(1+\gamma_i^*) \; ,  
\end{equation}
in which $\gamma_i^*$ is defined as
\begin{equation}
    \begin{aligned}
        \gamma_i^* = &\underset{\Bw_i(k)}{\max} \; \underset{k\in \CX_i}{\min} \; \textrm{SINR}_k \quad s.t. \quad \sum_{k\in \CX_i}||\Bw_{i}(k)||^2\leq P_T \; ,
    \end{aligned}
\end{equation}
where $\textrm{SINR}_k$ is the received SINR at user $k$ and $P_T$ is the available transmission power. We discuss this optimization problem in more details in the next section. 
%\footnote{While information theoretic in nature}
%CAN WE JUST REMOVE THIS. THIS IS WELL KNOWN THINGS. 
%It should be noted that the multicast rate in \eqref{eq:Ri} is an upper bound for the achievable rate of the proposed scheme. This is because the capacity expressions in information theory are based on the assumption that the lengths of the codewords approach infinity, while here we assume limited-length time intervals. Although this upper bound can never be exactly realized in practice, using state of the art coding techniques such as turbo and polar codes~\cite{berrou1996near,arikan2009channel}, one can approximately achieve the information theoretic bound with practical codeword block sizes at low to moderate SNR levels. In this sense, the provided multicast rate expression is a practically relevant upper bound for the achievable rate.

%As will be discussed later, our scheme achieves the optimal DoF without requiring a MAC (multiple access channel) receiver at any user. This enables optimized MMSE-type beamformers to be applied with low computational complexity, even for networks with a large number of users. The result is a low-subpacketization coded caching scheme with appropriate performance at the finite-SNR regime. 

As suggested before, we will consider that each transmission vector serves $|\CX_i|=t+\alpha$ users, where $\alpha\leq L$ is the multiplexing gain and is treated as a parameter of choice that can be tuned to obtain a better rate performance at finite-SNR. This $\alpha$ represents the number of independent streams in each transmission, and hence, reducing it implies sacrificing some spatial multiplexing gain for the purpose of increasing the beamforming gain, which can help reduce the worst-user effect in the finite-SNR regime. As we discuss later on, this same $\alpha$ can also be calibrated to control subpacketization and beamformer design complexity.

%% file: LinearSubpacks.tex
In this section, we present our low-complexity high performance cyclic caching scheme, which can be applied to any MISO setup in which $\alpha \ge t$.\footnote{For setups with $t > \alpha$, one can use the coded caching scheme presented in~\cite{lampiris2018adding} for reduced subpacketization.} 
%We first provide the mathematical representation of a basic scheme that requires subpacketization of $K(t+\alpha)$ to achieve the sum DoF of $t+\alpha$. Next, we clarify the procedure of the proposed scheme with the help of an illustrative graphical example. Finally, we use a user-grouping approach to modify the basic scheme in order to further reduce the required subpacketization by a factor of $\phi_{K,t,\alpha}^2$, where $\phi_{K,t,\alpha} = gcd(K,t,\alpha)$. 
The following theorem summarizes the DoF and subpacketization performance of the scheme:
%\vspace{-2pt}
\begin{thm}
\label{thm:main_contribution}
For the large MISO broadcast setup with $t \le \alpha \le L$, the sum DoF of $t+\alpha$ is achievable with a subpacketization 
\begin{equation}
\label{eq:min_subpack}
%    \frac{K(t+\alpha)}{\phi_{K,t,\alpha}^2}.
    \frac{K(t+\alpha)}{\big(gcd(K,t,\alpha)\big)^2} \; .
\end{equation}
\end{thm}
\begin{proof}
The proof is found in this current section, where we present the designed cyclic caching scheme and show that it employs the above subpacketization to achieve the sum DoF of $t+\alpha$. 
\end{proof}

In what follows, we first introduce a cache placement algorithm in Section~\ref{section:lin_placement}, which is based on a well-defined placement matrix and requires each file to be split into $K(t+\alpha)$ smaller parts (subpackets). In Sections~\ref{section:lin_delivery} and~\ref{section:lin_decoding}, we explain the delivery phase, in which the missing data parts are delivered to the requesting users with $I=K(K-t)$ multicast transmission vectors, each serving $t+\alpha$ subpackets to $t+\alpha$ different users. In Section~\ref{section:graphical_rep}, using an example network, we show that this delivery algorithm follows a simple graphical representation that involves circular shifts of two vectors over a tabular structure. Overall, in Sections~\ref{section:lin_placement} to~\ref{section:graphical_rep}, we present a scheme that satisfies all the requests with multicast transmissions that always contain $t+\alpha$ subpackets, implying a DoF of $t+\alpha$ with subpacketization $K(t+\alpha)$. Finally, in Section~\ref{section:Emanuele_part} we show 
%, which is the same value as the nominator of~\eqref{eq:min_subpack}. Following the discussions in Sections~\ref{section:lin_placement} to~\ref{section:graphical_rep}, in Section~\ref{section:Emanuele_part}, we show 
that by properly applying a user-grouping technique, subpacketization is further reduced by a factor of $\big(gcd(K,t,\alpha)\big)^2$, without any DoF loss. %This reduction is possible with a user grouping technique, which enables the original network to be mapped into a smaller \textit{virtual} network. We show that applying the coded caching scheme presented in Sections~\ref{section:lin_placement} to~\ref{section:graphical_rep} to the virtual network, and then \textit{elevating} the resulting cache placement and delivery procedures to be applicable to the original network, it is possible to reduce the required subpacketization by a factor of $(gcd(K,t,\alpha))^2$, while maintaining the DoF. %The reduction factor of $\phi_{K,t,\alpha}^2$ is the same value in the denominator of~\eqref{eq:min_subpack}, and hence, this completes the proof.
%The scheme achieving the optimal sum DoF with the reduced subpacketization stated in the theorem is presented later in this section, and the optimality of the achievable DoF can be proved by means of the lower bound in \cite{lampiris2018resolving}.

\vspace{-2pt}
\begin{rem}
\label{rem:optimal_dof}
When $\alpha=L$, the achieved DoF $t+L$ is exactly optimal under the assumption of one-shot linear schemes and uncoded placement (cf.~\cite{lampiris2018resolving}).
\end{rem}

%Theorem~\ref{thm:main_contribution} follows directly from the novel cache placement and content delivery algorithms presented in Sections~\ref{section:lin_placement}, \ref{section:lin_delivery}, and~\ref{section:Emanuele_part}. 

%takes on integer values, and we also note According to the definition of the greatest common divisor, all the three parameters $K$, $t$ and $\alpha$ are divisible by $\phi_{K,t,\alpha}$, and hence, the subpacketization value in~\eqref{eq:min_subpack} is an integer. Moreover, from~\eqref{eq:min_subpack} it is clear that for a fixed $t$ value, the subpacketization requirement of the  cyclic caching scheme scales linearly with the number of users $K$.
%THIS IS TRUE ONLY WITHOUT THE DENOMINATOR
%\footnote{If $t$ scales with $K$, the growth is quadratic, which is still much less than the exponential growth of other well-known schemes such as~\cite{shariatpanahi2018physical}.} 
%\end{rem}
We note that, for fixed $t$ and $L$, the above integer subpacketization scales linearly with $K$. This allows applying coded caching in larger networks, and entails the benefit of a reduced number of necessary transmissions which in turn implies a reduced number of beamformer design problems that need to be solved. 
As a quick comparison, if $K=20,t=4,L=\alpha = 8$, the 
proposed scheme requires subpacketization of 15, while the schemes in~\cite{mohajer2020miso} and~\cite{shariatpanahi2018physical} respectively require (approximately) $5 \times 10^3$ and $3 \times 10^7$ subpackets.
%current subpacketization of $20(12)/(gcd(8,4,8))^2 = 15$, is smaller than the subpacketization $4845$ from~\cite{mohajer2020miso} and much smaller than the subpacketization of about $3 \times 10^7$ required by the scheme in~\cite{shariatpanahi2018physical}. 
More comparisons are provided in Section~\ref{section:sim_results}.

%The result is a high-performance, low-subpacketization scheme, which is nicely applicable to large networks. A smaller subpacketization also generally implies a smaller number of transmissions, thus reducing the number of beamformer design problems to be solved. 

%As a quick comparison, consider a moderate-sized network of $K=20$ users, in which the transmitter is equipped with $L=8$ antennas and the coded caching gain is $t=4$. Assuming the antenna array is used to achieve the maximum spatial multiplexing gain of $\alpha=8$, our scheme achieves the optimal sum DoF of $4+8=12$ with a subpacketization requirement of 15. In comparison, the subpacketization requirement for the original multi-antenna scheme in~\cite{shariatpanahi2018physical} is larger than $3 \times 10^7$; and for the recent scheme in~\cite{mohajer2020miso} it is 4845. 

%\vspace{-5pt}
\subsection{Cache Placement}
\label{section:lin_placement}
For cache placement, we use a $K \times K$ binary placement matrix $\BV$ where the first row has $t$ consecutive 1's (other elements are zero) and each subsequent row is a circular shift of the previous row by one column. Given $\BV$, we split each file $W$ into $K$ packets $W_p$, $p \in [K]$, and each packet $W_p$ into $t+\alpha$ smaller subpackets $W_p^q$. Then for every $p,k\in [K]$, if $\BV[p,k] = 1$, $W_p^q$ is stored in the cache memory of user $k$, $\forall W \in \CF, q \in [t+\alpha]$.%For the sake of simplicity, we use the term \textit{packet} to refer to $W_p$, and \textit{subpacket} to refer to $W_p^q$, for some general $W \in \CF$, $p \in [K]$ and $q \in [t+\alpha]$.
\footnote{The placement matrix $\BV$ used in this paper is a special case of valid placement matrices introduced in~\cite{salehi2019subpacketization}.}

%\vspace{-2pt}
\begin{exmp}
\label{exmp:placement_matrix}
For a scenario of $K=6$, $t=2$, $\alpha=3$, $\BV$ is built as
%\vspace{2pt}
\begin{equation}\label{eq:V_example}
\BV = 
    \renewcommand\arraystretch{0.6} \begin{bmatrix}
    1 & 1 & 0 & 0 & 0 & 0 \\
    0 & 1 & 1 & 0 & 0 & 0 \\
    0 & 0 & 1 & 1 & 0 & 0 \\
    0 & 0 & 0 & 1 & 1 & 0 \\
    0 & 0 & 0 & 0 & 1 & 1 \\
    1 & 0 & 0 & 0 & 0& 1
    \end{bmatrix} \; ,
\end{equation}
%\vspace{2pt}
and the resulting subpacketization is $K \times (t+\alpha) = 6 \times (2+3) = 30$.
%If $\CF = \{ A,B,C,D,E,F \}$, we have
For example, the cache contents of users 1 and 2 can be found from~\eqref{eq:V_example} as
\begin{equation*}
    \begin{aligned}
    \CZ(1) = &\{ W_1^q, W_6^q \; ; \; \forall \, W \in \CF , \, q\in[5] \} \; , \\
    \CZ(2) = &\{ W_1^q, W_2^q \; ; \; \forall \, W \in \CF , \, q\in[5] \} \; . \\
    %&W_6^1, W_6^2, W_6^3, W_6^4, W_6^5 \; \mid \; W \in \CF\} \; ,
    %\CZ(1) = \{ &W_1^1, W_1^2, W_1^3, W_1^4, W_1^5, W_6^1, W_6^2, W_6^3, W_6^4, W_6^5, B_6^2, B_6^3, \\
    %&C_1^1, C_1^2, C_1^3, C_6^1, C_6^2, C_6^3, D_1^1, D_1^2, D_1^3, D_6^1, D_6^2, D_6^3,\\
    %&E_1^1, E_1^2, E_1^3, E_6^1, E_6^2, E_6^3, F_1^1, F_1^2, F_1^3, F_6^1, F_6^2, F_6^3 \} \; .
    \end{aligned}
\end{equation*}
The cache contents of users $3-6$ can be written accordingly.
\end{exmp}

\input{LinearS-Delivery}

%% file: LinearS-Delivery.tex
%\vspace{-5pt}
\subsection{Content Delivery}
\label{section:lin_delivery}
In cyclic caching, the content delivery phase consists of $K$ rounds, where in each round we build $K-t$ transmission vectors. Thus, the content delivery is completed after \rev{$I=K(K-t)$} transmissions. We use $\Bx^r_j$ to denote the transmission vector $j \in [K-t]$ at transmission round $r \in [K]$.\footnote{In the general transmission vector model~\eqref{eq:general_transmission_vector_model}, $\Bx_j^r$ corresponds to $\Bx_i$, $i=(r-1)(K-t)+j$.} By transmitting $\Bx^r_j$, useful data packets are delivered to a set of $t+\alpha$ users. We define the user index vector $\Bk_j^r$ to denote the set of users being targeted by $\Bx_j^r$, and the packet index vector $\Bp_j^r$ to contain the packet indices targeted for the users in $\Bk_j^r$. In other words, using $\Bx_j^r$, we transmit (part of) the packet $W_{\Bp_j^r[n]}(\Bk_j^r[n])$ to each user $\Bk_j^r[n]$,  $n=1, \ldots, t+\alpha$. Both $\Bk_j^r$ and $\Bp_j^r$ vectors are built recursively. Let us use $\%$ sign to denote the mod operator with an offset of one. It is defined as
\begin{equation}
    a \% b = ((a-1) \; \textrm{mod} \; b) + 1 \; ,
\end{equation}
such that $a \% a = a$ and $(a+b) \% a = b \% a$. Then, $\Bk_j^1$ and $\Bp_j^1$ are built as 
\begin{equation}
\label{eq:p1_k1_creation}
\begin{aligned}
    \Bk_j^1 &= \Big[ \; [1:t] \; ; \; \big( ([1:\alpha]+j-1)\%(K-t) \big) +t \; \Big] \; , \\
    %\Bp_j^1 &= \Big[ \; \Big( j \times \Be(t) \Big) + t \; ; \; \Be(\alpha) \; \Big] \; , needs-correction
    \Bp_j^1 &= \Big[ \big( (t+j-[1:t])\%(K-t) \big) + [1:t] \; ; \; \Be(\alpha) \Big] \; , \\
\end{aligned}
\end{equation}
where $\Be(m)$ is a vector of 1's with size $m$ (e.g., $\Be(3) = [1\;1\;1]$). For the next transmission rounds, i.e., $1 < r \le K$, we simply build $\Bk_j^r$ and $\Bp_j^r$, using $\Bk_j^1$ and $\Bp_j^1$, as
\begin{equation}
\label{eq:k_p_higher_rounds}
\begin{aligned}
    \Bk_j^r &= \big( \Bk_j^1 + r \big) \% K \; , \quad
    \Bp_j^r = \big( \Bp_j^1 + r \big) \% K \; .\\
\end{aligned}
\end{equation}
To gain a better insight into how $\Bk_j^r$ and $\Bp_j^r$ are built, in Section~\ref{section:graphical_rep}, we offer a simple graphical representation, which is based on circular shift operations over a tabular structure. In the following, we provide $\Bk_j^r$ and $\Bp_j^r$ vectors for the small network scenario given in Example~\ref{exmp:placement_matrix}.

%\vspace{-2pt}
\begin{exmp}
\label{exmp:user_packet_index_vectors}
In the scenario of Example~\ref{exmp:placement_matrix}, content delivery consists of six rounds, where at each round four transmission vectors are built. The user and packet index vectors for the first and second rounds are given as
\begin{equation}
\label{eq:example_first_round}
    \begin{aligned}
        \Bk_1^1 &= [1\;2\;3\;4\;5] \; , \quad \Bp_1^1 = [3\;3\;1\;1\;1] \; , \\
        \Bk_2^1 &= [1\;2\;4\;5\;6] \; , \quad \Bp_2^1 = [4\;4\;1\;1\;1] \; , \\
        \Bk_3^1 &= [1\;2\;5\;6\;3] \; , \quad \Bp_3^1 = [5\;5\;1\;1\;1] \; , \\
        \Bk_4^1 &= [1\;2\;6\;3\;4] \; , \quad \Bp_4^1 = [2\;6\;1\;1\;1] \; ,\\
    \end{aligned}
\end{equation}
and %for the second round we have
\begin{equation}
\label{eq:example_second_round}
    \begin{aligned}
        \Bk_1^2 &= [2\;3\;4\;5\;6] \; , \quad \Bp_1^2 = [4\;4\;2\;2\;2] \; , \\
        \Bk_2^2 &= [2\;3\;5\;6\;1] \; , \quad \Bp_2^2 = [5\;5\;2\;2\;2] \; , \\
        \Bk_3^2 &= [2\;3\;6\;1\;4] \; , \quad \Bp_3^2 = [6\;6\;2\;2\;2] \; , \\
        \Bk_4^2 &= [2\;3\;1\;4\;5] \; , \quad \Bp_4^2 = [3\;1\;2\;2\;2] \; , \\
    \end{aligned}
\end{equation}
%\vspace{2pt}
respectively. The user and packet index vectors for the other rounds are built similarly.
\end{exmp}

Two other variables are needed to build the transmission vector $\Bx_j^r$. First, we introduce the subpacket index $q(W,p)$, where $W \in \CF$ denotes a general file and $p \in [K]$ is the packet index. The subpacket index $q(W,p)$ indicates which subpacket of $W_p$ should be transmitted, the next time it is included in a transmission vector. For every $W \in \CF$ and $p \in [K]$, $q(W,p)$ is initialized to one, and incremented every time $W_p$ is included in a transmission vector. For notational simplicity, here we use
\begin{equation}
\label{eq:q_equivalency}
    q_j^r(n) \coloneqq q(W(\Bk_j^r[n]),\Bp_j^r[n]) \; .
\end{equation}
%to indicate which subpacket of each packet $W_{\Bp_j^r[n]}(\Bk_j^r[n])$ is transmitted in $\Bx_r^j$.  
Second, we define the \textit{interference indicator} set $\CR_j^r(n)$ as the set of users at which $W_{\Bp_j^r[n]}^{q_j^r(n)}(\Bk_j^r[n])$ should be suppressed by beamforming.\footnote{If zero-forcing beamformers are used, $\CR_j^r(n)$ denotes the set of users at which $W_{\Bp_j^r[n]}^{q_j^r(n)}(\Bk_j^r[n])$ should be nulled-out by beamforming.} $\CR_j^r(n)$ has $\alpha-1$ elements and is built as
\begin{equation}
\label{eq:zero_force_set_definition}
    \CR_j^r(n) = \Big\{k \in \Bk_j^r \backslash \Bk_j^r[n] \;\; \mid \;\; \BV\big[\Bp_j^r[n],k\big] = 0 \Big\} \; .
\end{equation}
Finally, the transmission vectors are built as:\footnote{The general transmission vector model in~\eqref{eq:general_transmission_vector_model} is equivalent to~\eqref{eq:generalTXvector} via the following index mapping:
\begin{equation*}
\begin{aligned}
    k \rightarrow \Bk_j^r[n] \; , \qquad &\qquad \qquad \; \; \Bw_i \rightarrow \Bw_{\CR_j^r(n)} \; , \\
    \CX_i \rightarrow \bigcup_{n \in [t+\alpha]} \{ \Bk_j^r[n]\} \; , &\qquad X_i(k) \rightarrow W_{\Bp_j^r[n]}^{q_j^r(n)}(\Bk_j^r[n]) \; .
\end{aligned}
\end{equation*}
}
\begin{equation}\label{eq:generalTXvector}
    \Bx_j^r = \sum_{n=1}^{t+\alpha} \Bw_{\CR_j^r(n)} W_{\Bp_j^r[n]}^{q_j^r(n)}(\Bk_j^r[n]) \; .
\end{equation}

%\vspace{-2pt}
\begin{exmp}
\label{exmp:R_sets}
%Consider the network in Example~\ref{exmp:placement_matrix}, for which the user and packet index vectors are provided in Example~\ref{exmp:user_packet_index_vectors}. 
For the network considered in Examples~\ref{exmp:placement_matrix} and~\ref{exmp:user_packet_index_vectors}, the interference indicator sets for the first transmission round are built as
\begin{equation}\label{eq:example3}
    \begin{aligned}
        \CR_1^1(1) = \{ 2,5 \} ,& \quad
        \CR_1^1(2) = \{ 1,5 \} , \quad
        \CR_1^1(3) = \{ 4,5 \} , \\ 
        \CR_1^1(4) &= \{ 3,5 \} , \quad
        \CR_1^1(5) = \{ 4,4 \} . \\ 
    \end{aligned}
\end{equation}
%Clearly, as $\alpha=3$, each set consists of exactly two elements.
\end{exmp}

%\vspace{-5pt}
\subsection{Decoding at the Receiver}
\label{section:lin_decoding}
During time interval $i$, every user $k \in \CX_i$ receives
\begin{equation}
\begin{aligned}
    y_i(k) = &\Bh_{k}^H \Bw_{i}(k) X_{i}(k) \\
        &+ \sum_{\hat{k} \in \CX_i \backslash \{k\}} \Bh_{k}^H \Bw_{i}(\hat{k}) X_{i}(\hat{k}) + w_i (k) \; ,
\end{aligned}
\end{equation}
where the first term is the intended codeword and the latter two terms indicate the interference and noise, respectively. Assume $\hat{k} \coloneqq \Bk_{\hat{j}}^{\hat{r}}[\hat{n}]$, for some $\hat{j},\hat{r},\hat{n}$. Defining $p(\hat{k}) \coloneqq \Bp_{\hat{j}}^{\hat{r}}[\hat{n}]$, for every element in the interference term only one of the following options is possible:
\begin{enumerate}
    \item $\BV[p(\hat{k}),k] = 1$ indicates $X_i(\hat{k})$ is in the cache memory of user $k$, and hence, $\Bh_{k}^H \Bw_{i}(\hat{k}) X_{i}(\hat{k})$ can be reconstructed and removed  from $y_i(k)$;
    \item $\BV[p(\hat{k}),k] = 0$ indicates that $\hat{k}$ is in the interference indicator set associated with $X_i(k)$ as defined in~\eqref{eq:zero_force_set_definition}, and hence, $X_i(\hat{k})$ is suppressed at user $k$ by transmit beamforming.
\end{enumerate}
In both cases, the interference due to $X_i(\hat{k})$ can be controlled and/or completely removed at user $k$. Since $|\CX_i| = t+ \alpha$, the proposed scheme allows for serving $t+\alpha$ users in parallel during each transmission interval. The following example clarifies the decoding procedure for a single transmission in a small network. 
A more detailed explanation is provided in Appendix~\ref{section:ICCvalidity}.

\begin{exmp}
\label{exmp:decoding_x_1}
Consider the network in Example~\ref{exmp:placement_matrix}, for which the user and packet index vectors are provided in Example~\ref{exmp:user_packet_index_vectors} and the interference indicator sets are presented in Example~\ref{exmp:R_sets}. Let us assume the demand set is $\CD = \{A,B,C,D,E,F \}$. Then, following~\eqref{eq:generalTXvector}, the first transmission vector in the first round is built as
\begin{equation}
\label{eq:the_first_ever_trans_vector}
\begin{aligned}
    \Bx_1^1 = & \ \Bw_{2,5} A_3^1 + \Bw_{1,5} B_3^1 %\\
    %&
    + \Bw_{4,5} C_1^1  + \Bw_{3,5} D_1^1 + \Bw_{3,4} E_1^1 \; ,
\end{aligned}
\end{equation}
where the brackets of the interference indicator sets are dropped for notation simplicity. 
%For better clarification, let us assume zero-forcing beamformers are used here. Then, 
After $\Bx_1^1$ is transmitted, user~1 receives
\begin{equation}
\label{eq:first_transmission}
    \begin{aligned}
        y_1^1(1) = & \ \Bh_1^H \Bw_{2,5} A_3^1 + \underline{ \underline{ \Bh_1^H \Bw_{1,5} B_3^1 } }  +  \underline{ \Bh_1^H \Bw_{4,5} C_1^1}  \\
        &+ \underline{\Bh_1^H \Bw_{3,5} D_1^1 } + \underline{\Bh_1^H \Bw_{3,4} E_1^1 } + w^1_1(1) \; ,
    \end{aligned}
\end{equation}
where the single- and double-underlined terms indicate the interference. From Example~\ref{exmp:placement_matrix}, $C_1^1$, $D_1^1$, and $E_1^1$ are available in the cache memory of user~1, and hence, all the single-underlined terms can be reconstructed and removed from the received signal. On the other hand, following the definition of the interference indicator sets, the double-underlined term (containing $B_3^1$) is also suppressed at user~1 with the help of the beamforming vectors. As a result, user~1 can decode $A_3^1$ with controlled interference. Similarly, users $2 - 5$ can decode $B_3^1$, $C_1^1$, $D_1^1$ and $E_1^1$, respectively. In Table~\ref{tab:deconding_x_1}, we have summarized how different users decode $\Bx_1^1$ and extract their requested data.
%\vspace{-5pt}
\begin{table}[t]
    \centering
    \resizebox{0.97\columnwidth}{!}{%
    \begin{tabular}{|c|c|c|c|c|}
         \multicolumn{5}{c}{Transmission vector: $\Bx_1^1 = \Bw_{2,5} A_3^1 + \Bw_{1,5} B_3^1 + \Bw_{4,5} C_1^1  + \Bw_{3,5} D_1^1 + \Bw_{3,4} E_1^1$} \\
        \hline
        \hline
         User & \begin{tabular}{@{}c@{}}Available \\in cache\end{tabular}  & \begin{tabular}{@{}c@{}}Supp. by \\ beamformer\end{tabular}  & \begin{tabular}{@{}c@{}}Useful \\ data\end{tabular}  & SINR \\
         \hline
         1 & $C_1^1, D_1^1, E_1^1$ & $B_3^1$ & $A_3^1$ & $\frac{|\Bh_1^H \Bw_{2,5}|^2}{|\Bh_1^H \Bw_{1,5}|^2 + N_0}$ \\
         \hline
         2 & $C_1^1, D_1^1, E_1^1$ & $A_3^1$ & $B_3^1$ & $\frac{|\Bh_12^H \Bw_{1,5}|^2}{|\Bh_2^H \Bw_{2,5}|^2 + N_0}$ \\
         \hline
         3 & $A_3^1,B_3^1$ & $D_1^1,E_1^1$ & $C_1^1$ & $\frac{|\Bh_3^H \Bw_{4,5}|^2}{|\Bh_3^H \Bw_{3,5}|^2 + |\Bh_3^H \Bw_{3,4}|^2 + N_0}$ \\
         \hline
         4 & $A_3^1,B_3^1$ & $C_1^1,E_1^1$ & $D_1^1$ & $\frac{|\Bh_4^H \Bw_{3,5}|^2}{|\Bh_4^H \Bw_{4,5}|^2 + |\Bh_4^H \Bw_{3,4}|^2 + N_0}$ \\
         \hline
         5 & $A_3^1,B_3^1$ & $C_1^1,D_1^1$ & $E_1^1$ & $\frac{|\Bh_5^H \Bw_{3,4}|^2}{|\Bh_5^H \Bw_{4,5}|^2 + |\Bh_5^H \Bw_{3,5}|^2 + N_0}$ \\
         \hline
         6 & $-$ & $-$ & $-$ & $-$ \\
         \hline
    \end{tabular}%
    }
    \caption{Decoding process for $\Bx_1^1$ at different network users, for Example~\ref{exmp:decoding_x_1}.}
    \label{tab:deconding_x_1}
\end{table}
\end{exmp}

%\vspace{-25pt}
\subsection{A Graphical Example}
\label{section:graphical_rep}
For further clarification, we describe the operation of the cyclic caching scheme for the network setup in Example~\ref{exmp:placement_matrix}, using a graphical representation of the placement matrix $\BV$ in Figure~\ref{fig:graph_exmp}. In this figure, each column represents a user, and each row denotes a packet index. For example, the first column represents user one, and the first row stands for the first packet of all files, i.e., $W_1^q$, $\forall W \in \CF, q \in [t+\alpha]$. Lightly shaded entries indicate packets that are cached at the user. For example, $W_1^q,W_6^q$ are stored at user~1, $\forall W \in \CF, q \in [t+\alpha]$. %Clearly, the cache placement indicated by Figure~\ref{fig:graph_exmp} is equivalent to the placement matrix $\BV$ provided in Example~\ref{exmp:placement_matrix}.

%\vspace{-10pt}
\input{Figures/Example_Tabular}
%\vspace{-15pt}

In the subsequent figures, we use darkly shaded entries to indicate which packet indices of the requested files are sent during each transmission. The column and row indices of these darkly shaded entries are extracted from the user and packet index vectors. For our example network, the user and packet index vectors for the first and second transmission rounds are provided in~\eqref{eq:example_first_round} and~\eqref{eq:example_second_round}, and their graphical representations are depicted in Figures~\ref{fig:graph_exmp_r1c1} and~\ref{fig:graph_exmp_r2c2}. %Such vectors for the first and second transmission rounds are provided in Example~\ref{exmp:user_packet_index_vectors}. 
%Based on these vectors, the graphical representation of the transmission vectors in the first and second transmission rounds are depicted in Figures~\ref{fig:graph_exmp_r1c1} and~\ref{fig:graph_exmp_r2c2}. 
For example, Fig.~\ref{fig:sub1} corresponds to the first transmission of the first round, where users $\Bk_1^1=[1,2,3,4,5]$ receive subpackets of packets indicated by $\Bp_1^1=[3, 3, 1, 1, 1]$. 
For simplicity, let us assume the demand set is $\CD = \{A,B,C,D,E,F \}$. Then, users 1-5 receive $A_3^1$, $B_3^1$, $C_1^1$, $D_1^1$ and $E_1^1$, respectively.
%in~\eqref{eq:example_first_round}, i.e., $W_3(1)$, $W_3(2)$, $W_1(3)$, $W_1(4)$ and $W_1(5)$, respectively.

\input{Figures/ICCExample-R1C1}
%\vspace{-25pt}
\input{Figures/ICCExample-R2C2}
%\vspace{-10pt}

The transmission vectors can be easily reconstructed using the graphical representations. For example, Fig.~\ref{fig:graph_exmp_r1c1} implies that the following transmission vectors are generated in the first round: 
\begin{equation*}
\label{eq:tx1}
    \begin{aligned}
        \Bx_1^1 = & \Bw_{2,5} A_3^1 + \Bw_{1,5} B_3^1 + \Bw_{4,5} C_1^1  + \Bw_{3,5} D_1^1 + \Bw_{3,4} E_1^1 \; , \\
        \Bx_2^1 = & \Bw_{2,6} A_4^1 + \Bw_{1,6} B_4^1 + \Bw_{5,6} D_1^2  + \Bw_{4,6} E_1^2 + \Bw_{4,5} F_1^1 \; , \\
        \Bx_3^1 = & \Bw_{2,3} A_5^1 + \Bw_{1,3} B_5^1 + \Bw_{6,3} E_1^3  + \Bw_{5,3} F_1^2 + \Bw_{5,6} C_1^2 \; , \\
        \Bx_4^1 = & \Bw_{4,6} A_2^1 + \Bw_{3,4} B_6^1 + \Bw_{3,4} F_1^3  + \Bw_{6,4} C_1^3 + \Bw_{6,3} D_1^3 \; , \\
    \end{aligned}
\end{equation*}
where the brackets of the interference indicator sets are dropped for notation simplicity. Note that according to~\eqref{eq:q_equivalency}, we have 
%\begin{equation*}
    $q_j^1(n) \equiv q(W(\Bk_j^1[n]),\Bp_j^1[n])$,
%\end{equation*}
and hence, $q_j^1(n)$ is incremented every time $W_{\Bp_j^1[n]}(\Bk_j^1[n])$ appears in a transmission vector. As a result, the subpacket index for the packet $C_1$ is incremented from one to three, as it has appeared in $\Bx_1^1$, $\Bx_3^1$ and $\Bx_4^1$, respectively.
%every subpacket index is initialized as $q_j^1(n)=1$ , and incremented by one whenever its respective packet is repeated in the transmission vectors. 

Following the same procedure, using Figure~\ref{fig:graph_exmp_r2c2}, for the second round we have
\begin{equation*}\label{eq:tx2}
    \begin{aligned}
        \Bx_1^2 = & \Bw_{3,6} B_4^2 + \Bw_{2,6} C_4^1 + \Bw_{5,6} D_2^1  + \Bw_{4,6} E_2^1 + \Bw_{4,5} F_2^1 \; , \\
        \Bx_2^2 = & \Bw_{3,1} B_5^2 + \Bw_{2,1} C_5^1 + \Bw_{6,1} E_2^2  + \Bw_{5,1} F_2^2 + \Bw_{5,6} A_2^2 \; , \\
        \Bx_3^2 = & \Bw_{3,4} B_6^2 + \Bw_{2,4} C_6^1 + \Bw_{1,4} F_2^3  + \Bw_{6,4} A_2^3 + \Bw_{6,1} D_2^2 \; , \\
        \Bx_4^2 = & \Bw_{5,1} B_3^2 + \Bw_{4,5} C_1^4 + \Bw_{4,5} A_2^4  + \Bw_{1,5} D_2^3 + \Bw_{1,4} E_2^3 \; . \\
    \end{aligned}
\end{equation*}
From Figures~\ref{fig:graph_exmp_r1c1} and~\ref{fig:graph_exmp_r2c2}, it can be seen that the transmissions vectors $\Bx_2^r$, $\Bx_3^r$, and $\Bx_4^r$ in round $r$ are built by circular shifts of the first transmission vector $\Bx_1^r$ over the non-shaded cells of the tabular grid and in two perpendicular directions. Specifically, the first two terms in $\Bx_1^r$ are shifted vertically, while the other three are shifted horizontally. This procedure is highlighted in sub-figures using wider border lines. Moreover, comparing Figures~\ref{fig:graph_exmp_r1c1} and~\ref{fig:graph_exmp_r2c2}, it is clear that the vectors in the second transmission round are diagonally shifted versions of the vectors in the first round. This property is the intuition behind the cyclic caching name and results from the recursive procedure in~\eqref{eq:k_p_higher_rounds}, where the mod operator is used to build $\Bk_j^r$ and $\Bp_j^r$ vectors for $r > 1$.

%\vspace{-10pt}
\subsection{Beamformer Design}
\label{section:lin_beamformer}
As discussed earlier, we use optimized beamformers to build the transmission vectors. These beamformers result in a better rate compared with zero-forcing, especially in the low-SNR regime, as they allow balancing the detrimental impact of noise and the inter-stream interference~\cite{tolli2017multi}. However, optimized beamformers may require non-convex optimization problems to be solved (due to interference from unwanted terms), making the problem computationally intractable even for moderate $ K $ values. Interestingly, in addition to requiring much-reduced subpacketization, cyclic caching also manages to eliminate the requirement of multicasting, thus enabling optimized beamformers to be designed with much less computational complexity.

As $t+\alpha$ users are served simultaneously by each transmission vector, symmetric rate maximization is equivalent to maximizing the worst user rate (among served users), which, in turn, is equivalent to maximizing the worst user SINR. Naturally, the unwanted terms canceled out using the local cache contents are \textit{not} considered interference in optimized SINR expressions. 
The optimized beamformer vectors for the $j$-th transmission in round $r$ can be found by solving the optimization problem
\begin{equation}
    \label{eq:optimizedBeamOP}
    \begin{aligned}
        \max_{\Bw_{\CR_j^r(n)}} & \min_{n\in [t+\alpha]} 
        \frac{ | \Bh_{\Bk_j^r[n]}^H \Bw_{\CR_j^r(n)} | ^2} {\sum\limits_{b \, : \, \CR_j^r[b] \ni \Bk_j^r[n]} | \Bh_{\Bk_j^r[n]}^H \Bw_{\CR_j^r(b)} | ^2 + N_0} \\
        &s.t. \quad \sum_{n \in [t+\alpha]} | \Bw_{\CR_j^r(n)} |^2 \leq P_T \; .
    \end{aligned}
\end{equation}

%\vspace{-2pt}
\begin{exmp}
\label{exmp:beamformer_problem}
Consider the network in Example~\ref{exmp:placement_matrix} and the transmitted signal vector $\Bx_1^1$ in~\eqref{eq:the_first_ever_trans_vector} for the first transmission of the round $r=1$. The optimized beamformers $ \Bw_{25}, \Bw_{15}, \Bw_{45}, \Bw_{35}, \Bw_{34}$ can be found by solving
\begin{equation*}
\label{eq:exampleoptim1}
\begin{aligned}
        \max_{\Bw_{\CR}} & \; \min \big\{ \frac{|\Bh_1^H\Bw_{25}|^2}{\lambda_{1}+N_0}, \frac{|\Bh_2^H\Bw_{15}|^2}{\lambda_{2}+N_0}, \\
        & \qquad \qquad 
        \frac{|\Bh_3^H\Bw_{45}|^2}{\lambda_{3}+N_0}, \frac{|\Bh_4^H\Bw_{35}|^2}{\lambda_{4}+N_0}, \frac{|\Bh_5^H\Bw_{34}|^2}{\lambda_{5}+N_0} \big\} \\
        s.t. \; \; & | \Bw_{25} | ^2 + \! | \Bw_{15} | ^2 + \! | \Bw_{45} | ^2 + \! | \Bw_{35} | ^2 + \! | \Bw_{34} | ^2 \leq \! P_T, \\
\end{aligned}
\end{equation*}
where $\lambda_{k}$ denotes the interference at user $k$, given as
\begin{equation*}
    \begin{aligned}
        &\lambda_{1} = |\Bh_1^H\Bw_{15}|^2 , \\
        &
        \lambda_{2} = |\Bh_2^H\Bw_{25}|^2 , \\
        &
        \lambda_{3} = |\Bh_3^H\Bw_{35}|^2+ |\Bh_3^H\Bw_{34}|^2 , \\
        &\lambda_{4} = |\Bh_4^H\Bw_{34}|^2+ |\Bh_4^H\Bw_{45}|^2 , \\
        &
        \lambda_{5} = |\Bh_5^H\Bw_{15}|^2+ |\Bh_5^H\Bw_{25}|^2+ |\Bh_5^H\Bw_{35}|^2+ |\Bh_5^H\Bw_{45}|^2 . \\
    \end{aligned}
\end{equation*}
%Clearly, the unwanted terms removed with the help of local cache contents are not appearing as interference.
\end{exmp}

In cyclic caching, as also demonstrated in Example~\ref{exmp:beamformer_problem}, the number of interfering messages experienced by each user does not need to be the same, in general. For the transmission vector $\Bx_1^1$ considered in Example~\ref{exmp:beamformer_problem}, users 1-5 experience 1, 1, 2, 2, and 4 interfering messages, respectively. This unevenness is an intrinsic characteristic of the proposed cyclic caching scheme, while each message is still suppressed at exactly $\alpha-1$ users (in Example~\ref{exmp:beamformer_problem}, there exist exactly $\alpha-1=2$ users in each interference indicator set).\footnote{%To the best of our knowledge, no scheme in the literature has this unevenness in the number of interfering messages at each intended user. 
We suspect that altering the placement scheme to remove this unevenness may improve the achievable rate due to the optimization problem's max-min structure. Removing this unevenness would require substantial changes to the scheme, however, and is part of the ongoing research.}

%As we saw in equations \eqref{eq:tx1} and \eqref{eq:tx2}, each transmission vector is constructed as the sum of $5$ precoded unicast messages intended to $5$ different users. While each such message is  suppressed at exactly $\alpha-1=2$ users, we notice that the number of interference terms (interfering messages) that each intended user experiences is not uniform. For example, for the transmission considered in example 5, there are $2$ users with only $1$ interfering message, other $2$ users with $2$ interfering messages and only $1$ user with $4$ interfering messages. This unevenness in the number of interference messages for each intended user, which varies from one transmission to another,  is an intrinsic characteristic of the proposed cache placement and delivery scheme. No scheme in the literature is known to have this kind of unevenness. We believe that removal of such unevenness could be beneficial for the optimized beamformers design. However, it currently remains an open question whether it is possible or not to modify the here proposed scheme in order to have the same number of interference terms per user. 

The optimized beamformer design problems tend to be non-convex in general and require iterative methods such as successive convex approximation (SCA) to be used~\cite{tolli2017multi}. Such methods can be computationally complex and make the implementation infeasible, especially for large networks. However, the unique unicasting nature of the cyclic caching transmission and the max-min SINR objective in~\eqref{eq:optimizedBeamOP} make the optimization problem quasi-convex~\cite{bengtsson2001optimal,wiesel2005linear}, and hence, allow us to use uplink-downlink duality to attain a simple iterative solution. As the beamformer design with uplink-downlink duality is thoroughly discussed in~\cite{wiesel2005linear}, here we only briefly review the required process. %Given the cyclic caching's optimized beamformer design problem in~\eqref{eq:optimizedBeamOP}, if we d
Denoting the normalized receive beamforming vectors for the dual uplink channel as $\Bv_{\CR_j^r(n)}$, $r \in [K]$, and $j \in [K-t]$, the uplink-downlink duality necessitates
\begin{equation}
  \sum_{n \in [t+\alpha]} \nu_n\|\Bv_{\CR_j^r(n)}\|^2 = \sum_{n \in [t+\alpha]} \|\Bw_{\CR_j^r(n)}\|^2 \; ,
\end{equation}
where the beamforming vectors $\Bv_{\CR_j^r(n)}$ and their power values $\nu_n$ can be found by maximizing the minimum of dual uplink SINR expressions:
\begin{equation}
    \label{eq:optimizedBeamOPUL}
    \begin{aligned}
        \max_{\Bv_{\CR_j^r(n)}, \nu_n} & \min_{n\in [t+\alpha]} 
        \gamma_n = \frac{ \nu_n | \Bh_{\Bk_j^r[n]}^H \Bv_{\CR_j^r(n)} | ^2} {\sum\limits_{b \, : \,\CR_j^r[b] \ni \Bk_j^r[n]} \nu_b  |\Bh_{\Bk_j^r[b]}^H \Bv_{\CR_j^r(n)} | ^2 + N_0}\\
        s.t. \quad & \sum_{n \in [t+\alpha]} \nu_n  \leq P_T, \quad  \|\Bv_{\CR_j^r(n)}\|^2=1\ \forall \ n \; .
    \end{aligned}
\end{equation}
%{\color{blue} Add a short explanation on the difference of (20) and (18)}
Note that the interference terms in the denominator of~\eqref{eq:optimizedBeamOPUL} have different indices compared with~\eqref{eq:optimizedBeamOP}. The dual uplink optimization problem in~\eqref{eq:optimizedBeamOPUL} is quasi-convex and can be solved optimally for the given unicast group~\cite{wiesel2005linear}. %Here we use the standard iterative solution, in which a target SINR value, denoted by $\bar{\gamma}$, is iteratively adjusted by bisection until the power constraint is met with desired convergence level $| P_T - \sum_{n\in[t+\alpha]} \nu_n |<\epsilon$.  %Assume $j$ and $r$ are fixed, and let us define $\nu_n = |\Bv_{\CR_j^r(n)}|$. 
%At each iteration, for every $n \in [t+\alpha]$, we first find $\nu_n$ for fixed $\Bv_{\CR_j^r(n)}$ such that $\gamma_n = \Bar{\gamma}$. Then, n
Here we use a standard iterative approach, where we adjust (e.g., by bisection) a target SINR value, denoted by $\bar{\gamma}$, until the power constraint is met with a desired convergence level $| P_T - \sum_{n\in[t+\alpha]} \nu_n |<\epsilon$. However, this requires finding the power coefficients $\nu_n$ resulting in the minimum total power $\sum_{n\in[t+\alpha]} \nu_n$, for a given target SINR value $\bar{\gamma}$. We use another internal iterative loop to address this issue. We first note that the (normalized) MMSE receiver $\Bv_{\CR_j^r(n)} = \frac{\bar{\Bv}}{\|\bar{\Bv}\|}$, where
%
%
%From~\cite{chiang2008power}, we know that the (normalized) MMSE receiver
\begin{equation}
\label{eq:optimal_mmse_filter}
    \bar{\Bv}= \frac{1}{\sqrt{\nu_n}} \Big(\sum_{b \, : \, \CR_j^r[b] \ni \Bk_j^r[n]} \nu_b \; \Bh_{\Bk_j^r[b]} \Bh_{\Bk_j^r[b]}^H + N_0 \BI \Big)^{-1} \Bh_{\Bk_j^r[n]}
\end{equation}
is the optimal RX beamformer solution for the dual uplink channel, for a fixed set of power values $\nu_n$. Plugging~\eqref{eq:optimal_mmse_filter} into~\eqref{eq:optimizedBeamOPUL}, we can write the uplink SINR compactly as
\begin{equation}
%\label{eq:dual_SINR}
    \gamma_n = \nu_n \Bh_{\Bk_j^r[n]}^H \Big(\sum_{b \, : \, \CR_j^r[b] \ni \Bk_j^r[n]} \nu_b \; \Bh_{\Bk_j^r[b]} \Bh_{\Bk_j^r[b]}^H + N_0 \BI \Big)^{-1} \Bh_{\Bk_j^r[n]} \; .
\end{equation}
Now, 
%
%
%Moreover, at the optimal solution point for~\eqref{eq:optimizedBeamOPUL}, all the SINR values $\gamma_n$ are equal to the same value $\bar{\gamma}$. So, we can solve the max-min problem in~\eqref{eq:optimizedBeamOPUL} iteratively via minimum power beamforming design for fixed target SINR value $\bar{\gamma}$ use an standard iterative solution, where a target SINR value $\bar{\gamma}$ is iteratively adjusted by bisection until the power constraint is met with a desired convergence level $| P_T - \sum_{n\in[t+\alpha]} \nu_n |<\epsilon$. , this solution requires finding $\nu_n$ values, which can also be done using another iterative approach. To do so, we first plug~\eqref{eq:optimal_mmse_filter} into~\eqref{eq:optimizedBeamOPUL} to write the uplink SINR compactly as
%\begin{equation}
%\label{eq:dual_SINR}
%    \gamma_n = \nu_n \Bh_{\Bk_j^r[n]}^H \Big(\sum_{b \, : \, \CR_j^r[b] \ni \Bk_j^r[n]} \nu_b \; \Bh_{\Bk_j^r[b]} \Bh_{\Bk_j^r[b]}^H + N_0 \BI \Big)^{-1} \Bh_{\Bk_j^r[n]} \; ,
%\end{equation}
%and then, 
for a fixed $\bar{\gamma}$, we can iteratively solve for optimal $\nu_n$ values resulting in the minimum total power, using the joint update method in~\cite{wiesel2005linear}. This involves updating $\nu_n$ values via:
\begin{equation}
\label{eq:dual_SINR}
\begin{aligned}
    &\nu_n^{(m)} = \frac{\bar{\gamma}}{\gamma_n^{(m-1)}} \nu_n^{(m-1)} = \\
    &\frac{\bar{\gamma}} {\Bh_{\Bk_j^r[n]}^H \Big(\sum\limits_{b \, : \, \CR_j^r[b] \ni \Bk_j^r[n]} \nu_b^{(m-1)} \; \Bh_{\Bk_j^r[b]} \Bh_{\Bk_j^r[b]}^H + N_0 \BI \Big)^{-1} \Bh_{\Bk_j^r[n]}} \; ,
\end{aligned}
\end{equation}
until for every $n \in [t+\alpha]$, $\gamma_n = \bar{\gamma}$. Note that the convergence of the joint update method in~\eqref{eq:dual_SINR} can be proved by the standard interference function approach~\cite{yates1995framework}. 

So, in summary, we use an outer iterative loop to find the target SINR value $\bar{\gamma}$ for which the power constraint is met, and an internal iterative loop to find the optimal power coefficients for any given $\bar{\gamma}$.
%Summarizing the iterative approach, we first fix a target SINR value $\bar{\gamma}$ and find optimal $\nu_n$ values using~\eqref{eq:dual_SINR}, and then update $\bar{\gamma}$ via bisection until $| P_T - \sum_{n\in[t+\alpha]} \nu_n |<\epsilon$. 
After the outer loop is converged, we calculate $\Bv_{\CR_j^r(n)}$, and then we can find max-min optimal downlink beamformers using
%\begin{equation}
    $\Bw_{\CR_j^r(n)} = \rho_n \Bv_{\CR_j^r(n)}$,
%\end{equation}
where $\rho_n$ is the downlink power associated with $\Bw_{\CR_j^r(n)}$. To compute $\rho_n$, we first define $\Ba = [a_1 \; a_2 \; ... \; a_{t+\alpha}]$ and $\BD$ to be a diagonal matrix of elements $a_1$, ..., $a_{t+\alpha}$, where
\begin{equation}
    a_n = \frac{\gamma_n}{(1+\gamma_n) |\Bh_{\Bk_j^r(n)}^H\Bv_n|^2} \; .
\end{equation}
Then, we define $\BG$ as a matrix of elements $g_{n,b}$, $n,b \in [t+\alpha]$, where $g_{n,b} = |\Bh_{\CR_j^r(n)}^H \Bv_b|^2$ if either $b=n$ or $\CR_j^r[b] \ni \Bk_j^r[n]$, and $g_{n,b} = 0$ otherwise. Finally, defining ${\boldsymbol \rho} = [\rho_1 \; \rho_2 \; ... \; \rho_n]$, we have
\begin{equation}
    {\boldsymbol \rho}  = (\BI-\BD \BG)^{-1} N_0 \Ba \; .
\end{equation}

%% file: Figures/Example_Tabular.tex
\begin{figure}[ht]
    \centering
    \begin{subfigure}{.44\textwidth}
        \centering
        \begin{tikzpicture}[scale = 0.5]
            % grid
            \begin{scope}<+->;
            \draw[step=1cm,very thin,black!60] (0,0) grid (6,6);
            %\draw[pattern=crosshatch, pattern color=black!30, thin](5,0)rectangle(6,6);
            %\draw[pattern=crosshatch, pattern color=black!30, thin](3,0)rectangle(4,6);
            %\draw[very thick](0,0)rectangle(2,5);
            %\draw[very thick](2,5)rectangle(6,6);
            %\draw[very thick](0,0)to(0,6);
            %\draw[very thick](2,0)to(2,6);
            %\draw[very thick](0,5)to(6,5);
            %\draw[very thick](0,6)to(6,6);
            \end{scope}
            % function
            \begin{scope}
            \FillBlack{1}{1};
            \FillBlack{1}{2};
            \FillBlack{2}{2};
            \FillBlack{2}{3};
            \FillBlack{3}{3};
            \FillBlack{3}{4};
            \FillBlack{4}{4};
            \FillBlack{4}{5};
            \FillBlack{5}{5};
            \FillBlack{5}{6};
            \FillBlack{6}{6};
            \FillBlack{6}{1};
            \end{scope}
        \end{tikzpicture}
    \end{subfigure}
    \caption{Graphical illustration for Example 1.}
    \label{fig:graph_exmp}
\end{figure}
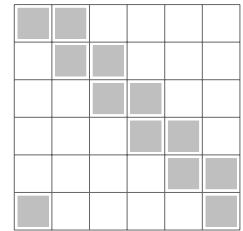

%% file: Figures/ICCExample-R1C1.tex
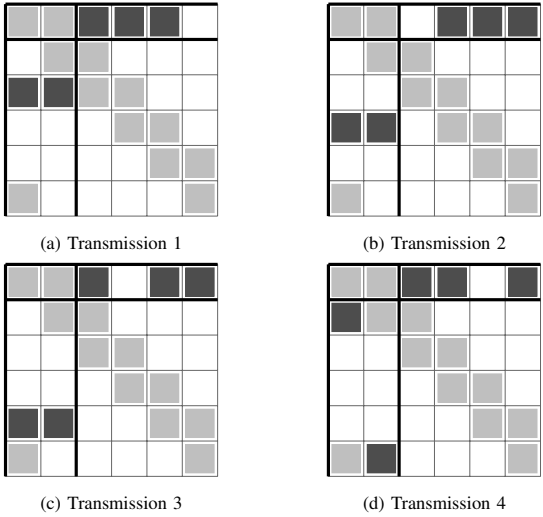
\begin{figure}[t]
    \centering
    \begin{subfigure}{.23\textwidth}
        \centering
        \begin{tikzpicture}[scale = 0.47]
            % grid
            \begin{scope}<+->;
            \draw[step=1cm,very thin,black!60] (0,0) grid (6,6);
            %\draw[pattern=crosshatch, pattern color=black!30, thin](5,0)rectangle(6,6);
            %\draw[pattern=crosshatch, pattern color=black!30, thin](3,0)rectangle(4,6);
            %\draw[very thick](0,0)rectangle(2,5);
            %\draw[very thick](2,5)rectangle(6,6);
            \draw[very thick](0,0)to(0,6);
            \draw[very thick](2,0)to(2,6);
            \draw[very thick](0,5)to(6,5);
            \draw[very thick](0,6)to(6,6);
            \end{scope}
            % function
            \begin{scope}
            \FillBlack{1}{1};
            \FillBlack{1}{2};
            \FillBlack{2}{2};
            \FillBlack{2}{3};
            \FillBlack{3}{3};
            \FillBlack{3}{4};
            \FillBlack{4}{4};
            \FillBlack{4}{5};
            \FillBlack{5}{5};
            \FillBlack{5}{6};
            \FillBlack{6}{6};
            \FillBlack{6}{1};
            \FillGray{1}{3};
            \FillGray{1}{4};
            \FillGray{1}{5};
            \FillGray{3}{1};
            \FillGray{3}{2};
            \end{scope}
        \end{tikzpicture}
        \caption{Transmission 1}
        \label{fig:sub1}
    \end{subfigure}
    \begin{subfigure}{.23\textwidth}
        \centering
        \begin{tikzpicture}[scale = 0.47]
            % grid
            \begin{scope}<+->;
            \draw[step=1cm,very thin,black!60] (0,0) grid (6,6);
            %\draw[pattern=crosshatch, pattern color=black!30, thin](2,0)rectangle(3,6);
            %\draw[pattern=crosshatch, pattern color=black!30, thin](3,0)rectangle(4,6);
            %\draw[very thick](0,0)rectangle(2,5);
            %\draw[very thick](2,5)rectangle(6,6);
            \draw[very thick](0,0)to(0,6);
            \draw[very thick](2,0)to(2,6);
            \draw[very thick](0,5)to(6,5);
            \draw[very thick](0,6)to(6,6);
            \end{scope}
            % function
            \begin{scope}
            \FillBlack{1}{1};
            \FillBlack{1}{2};
            \FillBlack{2}{2};
            \FillBlack{2}{3};
            \FillBlack{3}{3};
            \FillBlack{3}{4};
            \FillBlack{4}{4};
            \FillBlack{4}{5};
            \FillBlack{5}{5};
            \FillBlack{5}{6};
            \FillBlack{6}{6};
            \FillBlack{6}{1};
            \FillGray{1}{4};
            \FillGray{1}{5};
            \FillGray{1}{6};
            \FillGray{4}{1};
            \FillGray{4}{2};
            \end{scope}
        \end{tikzpicture}
        \caption{Transmission 2}
        \label{fig:sub2}
    \end{subfigure}
    \begin{subfigure}{.23\textwidth}
        \centering
        \begin{tikzpicture}[scale = 0.47]
            % grid
            \begin{scope}<+->;
            \draw[step=1cm,very thin,black!60] (0,0) grid (6,6);
            %\draw[pattern=crosshatch, pattern color=black!30, thin](3,0)rectangle(4,6);
            %\draw[very thick](0,0)rectangle(2,5);
            %\draw[very thick](2,5)rectangle(6,6);
            \draw[very thick](0,0)to(0,6);
            \draw[very thick](2,0)to(2,6);
            \draw[very thick](0,5)to(6,5);
            \draw[very thick](0,6)to(6,6);
            \end{scope}
            % function
            \begin{scope}
            \FillBlack{1}{1};
            \FillBlack{1}{2};
            \FillBlack{2}{2};
            \FillBlack{2}{3};
            \FillBlack{3}{3};
            \FillBlack{3}{4};
            \FillBlack{4}{4};
            \FillBlack{4}{5};
            \FillBlack{5}{5};
            \FillBlack{5}{6};
            \FillBlack{6}{6};
            \FillBlack{6}{1};
            \FillGray{1}{5};
            \FillGray{1}{6};
            \FillGray{1}{3};
            \FillGray{5}{1};
            \FillGray{5}{2};
            \end{scope}
        \end{tikzpicture}
        \caption{Transmission 3}
        \label{fig:sub3}
    \end{subfigure}
    \begin{subfigure}{.23\textwidth}
        \centering
        \begin{tikzpicture}[scale = 0.47]
            % grid
            \begin{scope}<+->;
            \draw[step=1cm,very thin,black!60] (0,0) grid (6,6);
            %\draw[pattern=crosshatch, pattern color=black!30, thin](4,0)rectangle(5,6);
            %\draw[pattern=crosshatch, pattern color=black!30, thin](3,0)rectangle(4,6);
            %\draw[very thick](0,0)rectangle(2,5);
            %\draw[very thick](2,5)rectangle(6,6);
            \draw[very thick](0,0)to(0,6);
            \draw[very thick](2,0)to(2,6);
            \draw[very thick](0,5)to(6,5);
            \draw[very thick](0,6)to(6,6);
            \end{scope}
            % function
            \begin{scope}
            \FillBlack{1}{1};
            \FillBlack{1}{2};
            \FillBlack{2}{2};
            \FillBlack{2}{3};
            \FillBlack{3}{3};
            \FillBlack{3}{4};
            \FillBlack{4}{4};
            \FillBlack{4}{5};
            \FillBlack{5}{5};
            \FillBlack{5}{6};
            \FillBlack{6}{6};
            \FillBlack{6}{1};
            \FillGray{1}{6};
            \FillGray{1}{3};
            \FillGray{1}{4};
            \FillGray{2}{1};
            \FillGray{6}{2};
            \end{scope}
        \end{tikzpicture}
        \caption{Transmission 4}
        \label{fig:sub4}
    \end{subfigure}
    \caption{Graphical illustration of the first round $r=1$.}
    \label{fig:graph_exmp_r1c1}
\end{figure}

%% file: Figures/ICCExample-R2C2.tex
\begin{figure}[t]
    \centering
    \begin{subfigure}{.22\textwidth}
        \centering
        \begin{tikzpicture}[scale = 0.47]
            % grid
            \begin{scope}<+->;
            \draw[step=1cm,very thin,black!60] (0,0) grid (6,6);
            %\draw[pattern=crosshatch, pattern color=black!30](5,0)rectangle(6,6);
            \draw[very thick](1,0)to(1,4)to(3,4)to(3,0);
            \draw[very thick](1,6)to(1,5)to(3,5)to(3,6);
            \draw[very thick](0,4)to(1,4)to(1,5)to(0,5);
            \draw[very thick](6,4)to(3,4)to(3,5)to(6,5);
            \end{scope}
            % function
            \begin{scope}
            \FillBlack{1}{1};
            \FillBlack{1}{2};
            \FillBlack{2}{2};
            \FillBlack{2}{3};
            \FillBlack{3}{3};
            \FillBlack{3}{4};
            \FillBlack{4}{4};
            \FillBlack{4}{5};
            \FillBlack{5}{5};
            \FillBlack{5}{6};
            \FillBlack{6}{6};
            \FillBlack{6}{1};
            \FillGray{2}{4};
            \FillGray{2}{5};
            \FillGray{2}{6};
            \FillGray{4}{2};
            \FillGray{4}{3};
            \end{scope}
        \end{tikzpicture}
        \caption{Transmission 1}
        \label{fig:sub12}
    \end{subfigure}
    \begin{subfigure}{.22\textwidth}
        \centering
        \begin{tikzpicture}[scale = 0.47]
            % grid
            \begin{scope}<+->;
            \draw[step=1cm,very thin,black!60] (0,0) grid (6,6);
            %\draw[pattern=crosshatch, pattern color=black!30](2,0)rectangle(3,6);
            \draw[very thick](1,0)to(1,4)to(3,4)to(3,0);
            \draw[very thick](1,6)to(1,5)to(3,5)to(3,6);
            \draw[very thick](0,4)to(1,4)to(1,5)to(0,5);
            \draw[very thick](6,4)to(3,4)to(3,5)to(6,5);
            \end{scope}
            % function
            \begin{scope}
            \FillBlack{1}{1};
            \FillBlack{1}{2};
            \FillBlack{2}{2};
            \FillBlack{2}{3};
            \FillBlack{3}{3};
            \FillBlack{3}{4};
            \FillBlack{4}{4};
            \FillBlack{4}{5};
            \FillBlack{5}{5};
            \FillBlack{5}{6};
            \FillBlack{6}{6};
            \FillBlack{6}{1};
            \FillGray{2}{5};
            \FillGray{2}{6};
            \FillGray{2}{1};
            \FillGray{5}{2};
            \FillGray{5}{3};
            \end{scope}
        \end{tikzpicture}
        \caption{Transmission 2}
        \label{fig:sub22}
    \end{subfigure}
    \begin{subfigure}{.22\textwidth}
        \centering
        \begin{tikzpicture}[scale = 0.47]
            % grid
            \begin{scope}<+->;
            \draw[step=1cm,very thin,black!60] (0,0) grid (6,6);
            %\draw[pattern=crosshatch, pattern color=black!30](3,0)rectangle(4,6);
            \draw[very thick](1,0)to(1,4)to(3,4)to(3,0);
            \draw[very thick](1,6)to(1,5)to(3,5)to(3,6);
            \draw[very thick](0,4)to(1,4)to(1,5)to(0,5);
            \draw[very thick](6,4)to(3,4)to(3,5)to(6,5);
            \end{scope}
            % function
            \begin{scope}
            \FillBlack{1}{1};
            \FillBlack{1}{2};
            \FillBlack{2}{2};
            \FillBlack{2}{3};
            \FillBlack{3}{3};
            \FillBlack{3}{4};
            \FillBlack{4}{4};
            \FillBlack{4}{5};
            \FillBlack{5}{5};
            \FillBlack{5}{6};
            \FillBlack{6}{6};
            \FillBlack{6}{1};
            \FillGray{2}{6};
            \FillGray{2}{1};
            \FillGray{2}{4};
            \FillGray{6}{2};
            \FillGray{6}{3};
            \end{scope}
        \end{tikzpicture}
        \caption{Transmission 3}
        \label{fig:sub32}
    \end{subfigure}
    \begin{subfigure}{.22\textwidth}
        \centering
        \begin{tikzpicture}[scale = 0.47]
            % grid
            \begin{scope}<+->;
            \draw[step=1cm,very thin,black!60] (0,0) grid (6,6);
            %\draw[pattern=crosshatch, pattern color=black!30](4,0)rectangle(5,6);
            \draw[very thick](1,0)to(1,4)to(3,4)to(3,0);
            \draw[very thick](1,6)to(1,5)to(3,5)to(3,6);
            \draw[very thick](0,4)to(1,4)to(1,5)to(0,5);
            \draw[very thick](6,4)to(3,4)to(3,5)to(6,5);
            \end{scope}
            % function
            \begin{scope}
            \FillBlack{1}{1};
            \FillBlack{1}{2};
            \FillBlack{2}{2};
            \FillBlack{2}{3};
            \FillBlack{3}{3};
            \FillBlack{3}{4};
            \FillBlack{4}{4};
            \FillBlack{4}{5};
            \FillBlack{5}{5};
            \FillBlack{5}{6};
            \FillBlack{6}{6};
            \FillBlack{6}{1};
            \FillGray{2}{1};
            \FillGray{2}{4};
            \FillGray{2}{5};
            \FillGray{3}{2};
            \FillGray{1}{3};
            \end{scope}
        \end{tikzpicture}
        \caption{Transmission 4}
        \label{fig:sub42}
    \end{subfigure}
    \caption{Graphical illustration of the second round $r=2$.}
    \label{fig:graph_exmp_r2c2}
\end{figure}
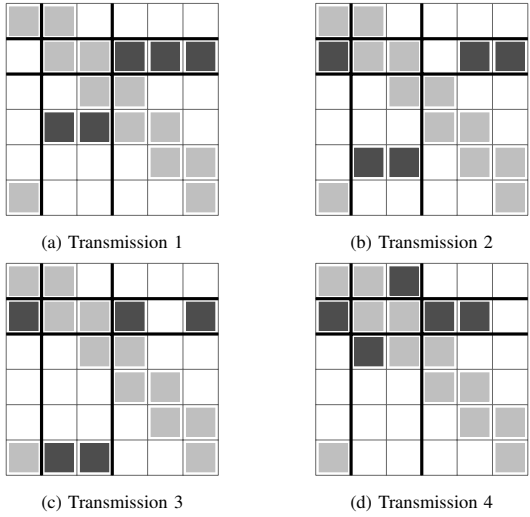

%% file: FurtherReduction-Summary.tex
Interestingly, with appropriate modifications, it is possible to further reduce the subpacketization requirement of cyclic caching by a factor of $(gcd(K,t,\alpha))^2$. This not only reduces the subpacketization requirement (and hence, the implementation complexity) considerably but also enables the subpacketization (and complexity) to be adjusted by tuning the $\alpha$ (and $K$) parameter. For notation simplicity, let us simply use $\phi=gcd(K,t,\alpha)$ to represent the reduction factor. To achieve the reduced subpacketization, we use a user grouping technique inspired by~\cite{lampiris2018adding}. The idea is to split users into groups of size $\phi$ and assume each group is equivalent to a \textit{virtual} user. Then, we consider a virtual network consisting of these virtual users, in which the coded caching and spatial multiplexing gains are $\frac{t}{\phi}$ and $\frac{\alpha}{\phi}$, respectively. Finally, we apply the coded caching scheme proposed in Sections~\ref{section:lin_placement} and~\ref{section:lin_delivery} to the virtual network, and \textit{elevate} the resulting cache placement and delivery schemes to be applicable in the original network. The elevation procedure, which is thoroughly explained in Appendix~\ref{secion:RED_details}, is designed to achieve the maximum DoF of $t+\alpha$ without any increase in the required subpacketization. As a result, the elevated scheme would require the same subpacketization as the scheme applied to the virtual network, which is $\frac{K}{\phi}\big( \frac{t}{\phi} + \frac{\alpha}{\phi} \big)$, as there exist $\frac{K}{\phi}$ virtual users in the virtual network. Here, we explain the proposed procedure with the help of one example.

%Clearly, as there exist $\frac{K}{\phi}$ virtual users, the subpacketization requirement is $\frac{K}{\phi}\big( \frac{t}{\phi} + \frac{\alpha}{\phi} \big)$.

%Then, we apply the currently proposed scheme to the network of the virtual users, assuming the coded caching gain is $\frac{t}{\phi}$ and the spatial multiplexing gain is $\frac{\alpha}{\phi}$. Finally, we \textit{elevate} the resulting cache placement and delivery schemes, to be applicable in the original network. The result is a coded caching scheme which achieves the total DoF of $t+\alpha$ in each transmission, and requires total subpacketization of $\frac{K(K-t)}{\phi^2}$. Here we explain the procedure with the help of a small example. Technical details are provided in Appendix~\ref{secion:RED_details}.

%with cache size of $Mf$ bits, the result would be a virtual network with $K'=\frac{K}{\phi}$ users and coded caching gain $t'=\frac{t}{\phi}$. Then, assuming the spatial DoF is $\alpha' = \frac{\alpha}{\phi}$, we design a coded caching scheme for the virtual network according to the algorithms presented Section~\ref{section:ICCPaper}. Clearly, the subpacketization requirement for the resulting scheme is $\frac{K}{\phi} \big( \frac{t}{\phi} + \frac{\alpha}{\phi} \big)$. Finally, we \textit{elevate} this coded caching scheme to be applicable to the original network, without any increase in the required subpacketization.

%\textcolor{blue}{the example needs enhancements. especially now the supporting text is shortened.}
\vspace{-5pt}
\begin{exmp}
\label{exmp:red_sub_base}
Assume a network scenario with $K=8$, $t=2$ and $\alpha=4$. In this case, $\phi = 2$ and the resulting virtual network has $K'=\frac{K}{\phi} = 4$ virtual users, coded caching gain of $t' = \frac{t}{\phi} = 1$ and spatial multiplexing gain of $\alpha' = \frac{\alpha}{\phi} = 2$. Assume the virtual users correspond to user groups $v_1 = \{1,2\}$, $v_2 = \{3,4\}$, $v_3 = \{5,6\}$ and $v_4 = \{7,8\}$. Applying the proposed cyclic caching scheme, the cache placement in the virtual network is
%\begin{equation}
    %\begin{aligned}
    $\CZ(v_{k'}) = \{ W_{k'}^q \mid W \in \CF, q \in [3] \}$, $\forall k' \in [4]$,
    %\end{aligned}
%\end{equation}
and the corresponding subpacketization requirement is $K'(t'+\alpha') = 12$. Data delivery is performed in four rounds, where three transmissions are done during each round. The first transmission vector in the first round is built as
\begin{equation}
\label{eq:emanuel_scheme_exmpl_trans_vector}
    {\Bx_1^1}' = \Bw_{v_3}' {W_2^1}'(v_1) + \Bw_{v_3}' {W_1^1}'(v_2) + \Bw_{v_2}' {W_1^1}'(v_3) \; ,
\end{equation}
and other transmission vectors are also built similarly. Now, to elevate the cache placement to be applicable to the original network, we simply bind the cache content of each user in the original network with its corresponding virtual user in the virtual network. So, the cache placement for the original network is
\begin{equation*}
\begin{aligned}
    \CZ(1) = \CZ(2) = \CZ(v_1) \; &= \Big\{ W_1^q \mid W \in \CF, q \in [3] \Big\} \; , \\
    \CZ(3) = \CZ(4) = \CZ(v_2) \; &= \Big\{ W_2^q \mid W \in \CF, q \in [3] \Big\} \; , \\
    \CZ(5) = \CZ(6) = \CZ(v_3) \; &= \Big\{ W_3^q \mid W \in \CF, q \in [3] \Big\} \; , \\
    \CZ(7) = \CZ(8) = \CZ(v_4) \; &= \Big\{ W_4^q \mid W \in \CF, q \in [3] \Big\} \; . \\
\end{aligned}
\end{equation*}
Elevating the delivery procedure is more complex. Let us consider the first transmission of the first round, as provided in~\eqref{eq:emanuel_scheme_exmpl_trans_vector}. The first term in the transmission vector, i.e., ${W_2^1}'(v_1)$, is suppressed at user $v_3$, which is equivalent to users $\{5,6\}$ in the original network. However, ${W_2^1}'(v_1)$ itself is combined of two subpackets destined to users 1 and 2, i.e., $W_2^1(1)$ and $W_2^1(2)$. The interference from these two subpackets should also be suppressed at users 2 and 1, which requires the interference indicator set for the first term to be $\{5,6 \} \cup \{2 \}$ and the one for the second term to be $\{5,6 \} \cup \{1 \}$. Following the same procedure for the second and third terms, the equivalent transmission vector for the original network will be
\begin{equation}
\begin{aligned}
    \Bx_1^1 &= \Bw_{5,6,2} W_2^1(1) + \Bw_{5,6,1} W_2^1(2) + \Bw_{5,6,4} W_1^1(3) \\
    &+ \Bw_{5,6,3} W_1^1(4) + \Bw_{3,4,6} W_1^1(5) + \Bw_{3,4,5} W_1^1(6)\;.
\end{aligned}
\end{equation}
Other transmissions are built similarly. Overall, the algorithm requires subpacketization of 12 and delivers data in 12 transmissions. In comparison, applying cyclic caching without user grouping requires subpacketization of $K(t+L) = 48$, and the number of transmissions would be 48. A graphical representation for the first transmission of the first round, for both virtual and original networks, is provided in Figure~\ref{fig:red_example}.
\vspace{-5pt}
\input{Figures/RED_Example}
\vspace{-15pt}
\end{exmp}

%% file: Figures/RED_Example.tex
\begin{figure}[ht]
    \centering
    \begin{subfigure}{.4\columnwidth}
        \centering
        \begin{tikzpicture}[scale = 0.5]
            % grid
            \begin{scope}<+->;
            \draw[step=1cm,very thin,black!60] (0,0) grid (4,4);
            %\draw[pattern=crosshatch, pattern color=black!30, thin](5,0)rectangle(6,6);
            %\draw[pattern=crosshatch, pattern color=black!30, thin](3,0)rectangle(4,6);
            %\draw[very thick](0,0)rectangle(2,5);
            %\draw[very thick](2,5)rectangle(6,6);
            \draw[very thick](0,0)to(0,4);
            \draw[very thick](1,0)to(1,4);
            \draw[very thick](2,0)to(2,4);
            \draw[very thick](3,0)to(3,4);
            \draw[very thick](4,0)to(4,4);
            \draw[very thick](0,0)to(4,0);
            \draw[very thick](0,1)to(4,1);
            \draw[very thick](0,2)to(4,2);
            \draw[very thick](0,3)to(4,3);
            \draw[very thick](0,4)to(4,4);
            \end{scope}
            % function
            \begin{scope}
            \FillBlackNf{1}{1};
            \FillBlackNf{2}{2};
            \FillBlackNf{3}{3};
            \FillBlackNf{4}{4};
            \FillBlackNf{3}{3};
            \FillGrayNf{1}{2};
            \FillGrayNf{1}{3};
            \FillGrayNf{2}{1};
            \end{scope}
        \end{tikzpicture}
        \caption{Virtual Network}
        \label{fig:sub_small}
    \end{subfigure}
    \begin{subfigure}{.55\columnwidth}
        \centering
        \begin{tikzpicture}[scale = 0.5]
            % grid
            \begin{scope}<+->;
            \draw[step=1cm,very thin,black!60] (0,0) grid (8,4);
            %\draw[pattern=crosshatch, pattern color=black!30, thin](5,0)rectangle(6,6);
            %\draw[pattern=crosshatch, pattern color=black!30, thin](3,0)rectangle(4,6);
            %\draw[very thick](0,0)rectangle(2,5);
            %\draw[very thick](2,5)rectangle(6,6);
            \draw[very thick](0,0)to(0,4);
            \draw[very thick](2,0)to(2,4);
            \draw[very thick](4,0)to(4,4);
            \draw[very thick](6,0)to(6,4);
            \draw[very thick](8,0)to(8,4);
            \draw[very thick](0,0)to(8,0);
            \draw[very thick](0,1)to(8,1);
            \draw[very thick](0,2)to(8,2);
            \draw[very thick](0,3)to(8,3);
            \draw[very thick](0,4)to(8,4);
            \end{scope}
            % function
            \begin{scope}
            \FillBlackNf{1}{1};
            \FillBlackNf{1}{2};
            \FillBlackNf{2}{3};
            \FillBlackNf{2}{4};
            \FillBlackNf{3}{5};
            \FillBlackNf{3}{6};
            \FillBlackNf{4}{7};
            \FillBlackNf{4}{8};
            \FillGrayNf{1}{3};
            \FillGrayNf{1}{4};
            \FillGrayNf{1}{5};
            \FillGrayNf{1}{6};
            \FillGrayNf{2}{1};
            \FillGrayNf{2}{2};
            \end{scope}
        \end{tikzpicture}
        \caption{Original Network}
        \label{fig:sub_big}
    \end{subfigure}
    \caption{Graphical illustration of transmission vector in Example~\ref{exmp:red_sub_base}}
    \label{fig:red_example}
\end{figure}
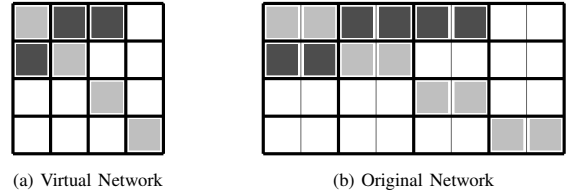

%% file: SimResults.tex
\subsection{Complexity Analysis}
For two main reasons, the subpacketization value is an important complexity indicator for any coded caching scheme. First, it indicates the number of smaller parts each file must be split into for the scheme to work correctly. As argued in~\cite{lampiris2018adding}, the exponentially growing subpacketization of conventional coded caching schemes can make their implementation infeasible, even for networks with a moderate number of users. Second, a scheme with smaller subpacketization \textit{generally} requires a smaller number of transmissions, and consequently, fewer beamformer design problems to be solved. For comparison, cyclic caching requires subpacketization and transmission count of $\frac{K(t+ \alpha)}{\phi_{K,t,\alpha}^2}$ and $\frac{K(K-t)}{\phi_{K,t,\alpha}^2}$, respectively. Both these numbers are considerably smaller than the original multi-antenna scheme in~\cite{shariatpanahi2018physical}, which requires subpacketization and transmission count of $\binom{K}{t} \binom{K-t-1}{L-1}$ and $\binom{K}{t+L}$, respectively.
% transmissions are needed in total. 
%As each transmission requires solving a separate beamformer design problem, cyclic caching has remarkably lower computation complexity than in~\cite{tolli2017multi}. 
%Overall, as the total DoF achieved during each transmission is the same for both schemes, the complexity per transmitted data unit is also much smaller in the proposed scheme.

\begin{table*}[t]
    \centering
    \begin{tabular}{|m{0.24\textwidth}|m{0.7\textwidth}|}
        \hline
         Reduced subpacketization & Cyclic caching enables achieving the sum-DoF of $t+\alpha$, with the reduced subpacketization of $\frac{K(t+\alpha)}{\phi_{K,t,\alpha}^2}$. \\
         \hline
         Reduced number of transmissions & With cyclic caching, delivery to all users is completed with only $\frac{K(K-t)}{\phi_{K,t,\alpha}^2}$ transmissions. \\
         \hline
         Relying only on unicasting & Cyclic caching relies on unicasting only. As a result, optimized MMSE-type beamformers can be implemented with much lower complexity using the uplink-downlink duality. This improves the performance at low- and mid-SNR, compared with ZF beamforming. \\
         \hline
         No MAC decoding & Cyclic caching removes the requirement of MAC decoding, thus eliminating the necessity of complex receiver structures such as successive interference cancellation (SIC). \\
         \hline
         Controlling the complexity & The subpacketization of cyclic caching can be controlled by tuning the $\phi_{K,t,\alpha}$ parameter, which is possible by adjusting $\alpha$ and also by considering a set of $K_f$ phantom users. \\
         \hline
    \end{tabular}
    \caption{Advantages of the Cyclic Caching Scheme}
    \label{tab:advantages}
\end{table*}

%\begin{table*}[t]
%\centering
%    \begin{minipage}{\textwidth}
%    \centering
%    \begin{tabular}{|c|c|c|c|c|c|}
%         \hline
%         & M-S & LIN & RED & L-E & M-B \\
%         \hline
%         Subpacketization & $O(K^t K^{L-1})$ & $O(K)$ & $O(K)$ & $O(K^{\nicefrac{t}{L}})$ & $O(K^t)$ \\
%         \hline
%         Transmission Count & $O(K^{t+L})$ & $O(K^2)$ & $O(K^2)$ & $O(K^{\nicefrac{t}{L}+1})$ & $O(K^{t+1})$ \\
%         \hline
%    \end{tabular}
%    \caption{Complexity order of different schemes when $t$ does not scale with $K$}
%    \label{tab:comp_noscale}
%    \end{minipage}
%
%    \begin{minipage}{\textwidth}
%    \centering
%    \begin{tabular}{|c|c|c|c|c|c|}
%         \hline
%         & M-S & LIN & RED & L-E & M-B \\
%         \hline
%         Subpacketization & $O(2^{K H(\gamma)} K^{L-1})$ & $O(K^2)$ & $O(K^2)$ & $O(2^{K \nicefrac{H(\gamma)}{L^2}})$ & $O(2^{K H(\gamma)})$ \\
%         \hline
%         Transmission Count & $O(2^{K H(\gamma)})$  & $O(K^2)$ & $O(K^2)$ & $O(2^{K \nicefrac{H(\gamma)}{L}})$ & $O(2^{K H(\gamma)})$  \\
%         \hline
%    \end{tabular}
%    \caption{Complexity order of different schemes when $t$ scales with $K$, $\gamma = \frac{M}{N}$, $H(\cdot)$ is the entropy function}
%    \label{tab:comp_scale}
%    \end{minipage}
%
%\end{table*}

\begin{table*}[t]
\centering
    \centering
    \resizebox{\textwidth}{!}{%
    \begin{tabular}{|c||c|c||c|c||c|c||c|c||c|c|}
         \hline
         Scheme & \multicolumn{2}{c||}{M-S} & \multicolumn{2}{c||}{LIN} & \multicolumn{2}{c||}{RED} & \multicolumn{2}{c||}{L-E} & \multicolumn{2}{c|}{M-B} \\
         \hline
         Scaling of $t$ & fixed & scaling & fixed & scaling & fixed & scaling & fixed & scaling & fixed & scaling \\
         \hline
         \hline
         Subpacketization & $O(K^t K^{L-1})$ & $O(2^{K H(\gamma)} K^{L-1})$ & $O(K)$ & $O(K^2)$ & $O(K)$ & $O(K^2)$ & $O(K^{\nicefrac{t}{L}})$ & $O(2^{K \nicefrac{H(\gamma)}{L^2}})$ & $O(K^t)$ & $O(2^{K H(\gamma)})$ \\
         \hline
         Trans. Count & $O(K^{t+L})$ & $O(2^{K H(\gamma)})$ & $O(K^2)$ & $O(K^2)$ & $O(K^2)$ & $O(K^2)$ & $O(K^{\nicefrac{t}{L}+1})$ & $O(2^{K \nicefrac{H(\gamma)}{L}})$ & $O(K^{t+1})$ & $O(2^{K H(\gamma)})$ \\
         \hline
    \end{tabular}%
    }
    \caption{Complexity order of different schemes when $t$ is fixed and when it scales with $K$, $\gamma = \frac{M}{N}$, $H(\cdot)$ is the entropy function}
    \label{tab:comp_scal_noscal}
\end{table*}

\begin{table*}[t]
    \centering
    
    \begin{minipage}{\textwidth}
    \centering
    \begin{tabular}{|c|c|c|c|c||c|c|c|c|c||c|c|c|c|c|}
        \cline{6-15}
        \multicolumn{5}{c||}{} & \multicolumn{5}{c||}{Required Subpacketization} & \multicolumn{5}{c|}{Required number of transmissions ($I$)} \\
        \hline
        $K$ & $t$ & $L$ & $\alpha$ & $K_f$ & M-S & LIN & RED & L-E & M-B &  M-S & LIN & RED & L-E & M-B \\
        \hline
        \hline
        8 & 2 & 5 & 2 & 0 & 140 & 32 & 8 & 4 & - & 70 & 48 & 12 & 6 & -  \\
        \hline
        8 & 2 & 5 & 4 & 0 & 280 & 48 & 12 & - & 28 & 28 & 48 & 12 & - & 28 \\
        \hline
        8 & 2 & 5 & 5 & 0 & 140 & 56 & 56 & - & - & 8 & 48 & 48 & - & - \\
        \hline
        30 & 4 & 8 & 4 & 0 & $> 10^7$ & 240 & 60 & - & -  & $> 10^6$ & 780 & 195 & - & -  \\
        \hline
        30 & 4 & 8 & 4 & 2 & $> 10^8$ & 256 & 16 & 8 & - & $> 10^7$ & 832 & 52 & 28 & - \\
        \hline
        30 & 4 & 8 & 6 & 0 & $> 10^{9}$ & 300 & 75 & - & $> 10^{4}$ & $> 10^{7}$ & 780 & 195 & - & $> 10^{4}$ \\
        \hline
        100 & 15 & 30 & 15 & 0 & $> 10^{32}$ & 3000 & 120 & - & - & $> 10^{32}$ & 8500 & 340 & - & - \\
        \hline
        100 & 15 & 30 & 15 & 5 & $> 10^{33}$ & 3150 & 14 & 7 & - & $> 10^{26}$ & 9450 & 42 & 21 & - \\
        \hline
        100 & 15 & 30 & 17 & 0 & $> 10^{34}$ & 3200 & 3200 & - & $> 10^{17}$  & $> 10^{26}$ & 8500 & 8500 & - & $> 10^{17}$ \\
        \hline
        400 & 50 & 200 & 100 & 0 & $> 10^{153}$ & $>10^{4}$ & 24 & - & - & $> 10^{113}$ & $>10^{5}$ & 56 & - & - \\
        \hline
    \end{tabular}
    \caption{Complexity comparison for some example network setups}
    \label{tab:subpack_comp}
    \end{minipage}
    
%    \begin{minipage}{\textwidth}
%    \centering
%    \begin{tabular}{|c|c|c|c|c||c|c|c|c|c|}
%        \hline
%        $K$ & $t$ & $L$ & $\alpha$ & $K_f$ & M-S & LIN & RED & L-E & M-B \\
%        \hline
%        \hline
%        8 & 2 & 5 & 2 & 0 & 70 & 48 & 12 & 6 & -  \\
%        \hline
%        8 & 2 & 5 & 4 & 0 & 28 & 48 & 12 & - & 28 \\
%        \hline
%        8 & 2 & 5 & 5 & 0 & 8 & 48 & 48 & - & - \\
%        \hline
%        30 & 4 & 8 & 4 & 0 & $> 10^6$ & 780 & 195 & - & -  \\
%        \hline
%        30 & 4 & 8 & 4 & 2 & $> 10^7$ & 832 & 52 & 28 & - \\
%        \hline
%        30 & 4 & 8 & 6 & 0 & $> 10^{7}$ & 780 & 195 & - & $> 10^{4}$ \\
%        \hline
%        100 & 15 & 30 & 15 & 0 & $> 10^{32}$ & 8500 & 340 & - & - \\
%        \hline
%        100 & 15 & 30 & 15 & 5 & $> 10^{26}$ & 9450 & 42 & 21 & - \\
%        \hline
%        100 & 15 & 30 & 17 & 0 & $> 10^{26}$ & 8500 & 8500 & - & $> 10^{17}$ \\
%        \hline
%        400 & 50 & 200 & 100 & 0 & $> 10^{113}$ & $>10^{5}$ & 56 & - & - \\
%        \hline
%    \end{tabular}
%    \caption{Required number of transmissions \rev{($I$)}}
%    \label{tab:trans_comp}
%    \end{minipage}
    
\end{table*}

In addition to the performance and complexity benefits of reduced subpacketization, cyclic caching also has the critical advantage of relying on unicasting only, unlike other traditional schemes that rely on high-order multicast messages (e.g., using XOR-ed codewords).  
%This is a result of the lack of bit-wise codeword creation (e.g. with the XOR operator), unlike the traditional coded caching schemes.
Of course, as discussed in~\cite{salehi2019subpacketization}, removing multicasting causes performance loss. However, it enables high-performance optimized (MMSE-type) beamformers to be applicable with low complexity, even for large networks with a large sum-DoF. %Optimized beamformers can improve the rate performance especially at the finite-SNR regime by balancing the detrimental impact of both noise and inter-stream interference~\cite{tolli2017multi}. 
In the following subsection, we provide simulation results for large networks with up to $K=100$ users, in which optimized beamformers are used. To the best of our knowledge, this is the first time a multi-antenna coded caching scheme has been applied with optimized beamformers to such a large network with a large sum-DoF.

%At the same time, the proposed scheme not only reduces the required subpacketization but also relies entirely on unicasting. This enables high-performance optimized MMSE-type beamformers to be applicable with low complexity, even for large networks with a large sum-DoF value. As discussed in~\cite{tolli2017multi}, optimized beamformers improve the rate performance especially at the finite-SNR regime, as they enable an interplay between nulling out and controlling the inter-stream interference. In the next subsection, we provide simulation results for large networks with $K=100$ users, in which optimized beamformers are used and the sum-DoF is $t+\alpha = 30$. To the best of our knowledge, this is the first time a multi-antenna coded caching scheme has been applied with optimized beamformers to such a large network with such large parameters.
From another perspective, cyclic caching also removes the requirement of decoding multiple data parts jointly at the same user during a single transmission. As a result, there is no need for complex receiver schemes such as successive interference cancellation (SIC). Cyclic caching requires that every user in the target group decodes only one message during every transmission, while in~\cite{shariatpanahi2018physical}, $\binom{t+L-1}{t}$ terms must be jointly decoded. 
%Even though, as shown in~\cite{tolli2017multi}, reducing t
Of course, this complexity reduction is also possible by splitting each transmission with overlapping multicast messages into multiple TDMA intervals, as shown in~\cite{tolli2017multi}. However, it comes with a substantial subpacketization increase.
%The receiver complexity reduction is possible also by splitting each transmission with overlapping multicast messages into multiple TDMA intervals as shown in~\cite{tolli2017multi}, but it comes with a substantial subpacketization increase. %Overall, reducing the subpacketization while also removing the MAC requirement is one of the novel aspects of cyclic caching.

Interestingly, cyclic caching also enables tuning $\alpha$ and $K$ parameters jointly to reduce both subpacketization and transmission count. %This is accomplished by tuning $\alpha$ and $K$ to maximize $\phi=gcd(K,t,\alpha)$, in order to reduce both subpacketization and transmission count. 
This reduction %in subpacketization and transmission count 
is useful especially for large networks with a large number of users $K$, for which the complexity of coded caching schemes is critically limiting their practical implementation~\cite{lampiris2018adding}. Selecting $\alpha$ to be smaller than the antenna count is straightforward and explained in~\cite{tolli2017multi}. However, in order to tune $K$, we consider a set of $K_f$ additional \textit{phantom} users and tune both $\alpha$ and $K_f$ to maximize $gcd(K+K_f,t,\alpha)$. %and apply cyclic caching assuming the number of users is $K+K_f$ (the values of $t$ and $\alpha$ remain unchanged). 
$K_f$ phantom users are then omitted during all the subsequent transmissions. %by assigning them zero desired information. %Similar to tuning $\alpha$, this can also result in reduced subpacketization (and transmission count). 
Of course, tuning either parameter comes with a DoF loss. For $\alpha$, this is not necessarily an issue, especially when the communication is at finite-SNR. In~\cite{tolli2017multi}, it is shown that by choosing $\alpha<L$, one can obtain an improved beamforming gain, which considerably improves the performance at the finite-SNR regime. The joint impact of $\alpha$ and $K_f$ tuning is studied through numerical simulations in Section~\ref{sec:sim_results}. 
%$\phi_{K,t,\alpha}$ can be controlled not only by selecting $\alpha$ but also by tuning the user count $K$. This is accomplished by considering a set of $K_f$ additional users, and applying the proposed scheme assuming the number of users is $K+K_f$ (the values of $t$ and $\alpha$ remain unchanged). The additional users are then excluded by assigning them zero desired information. This can result in reduced subpacketization (and transmission count), by altering $gcd(K+K_f,t,\alpha)$. This can be especially interesting in large networks, in which a small $K_f$ can result in a considerable reduction in the complexity. 

In principle, the reduced subpacketization scheme of Section~\ref{section:Emanuele_part} can be applied by splitting users into groups of size $Q>1$, such that $gcd(K,t,\alpha)$ is divisible by $Q$. This enables selecting several subpacketization levels, between $K(t+\alpha)$ and $\frac{K(t+\alpha)}{gcd(K,t,\alpha)^2}$. However, as we show later, the performance of cyclic caching is almost intact regardless of the selected subpacketization, and hence, it makes sense to select the largest possible $gcd(K,t,\alpha)$ value.

In Table~\ref{tab:advantages}, we have summarized the key advantages of the cyclic caching scheme. Moreover, in Table~\ref{tab:comp_scal_noscal}, we have compared the complexity order, in terms of both subpacketization requirement and the total number of transmissions, for the multi-antenna scheme of~\cite{shariatpanahi2018physical} (M-S), cyclic caching without user grouping (LIN), cyclic caching with user grouping (RED), the original group-based scheme in~\cite{lampiris2018adding} (L-E), and the recently proposed scheme in~\cite{mohajer2020miso} (M-B). In Table~\ref{tab:comp_scal_noscal}, the complexity order is provided for two cases where the global cache ratio $t = \frac{KM}{N}$ is fixed and when it scales with the number of users $K$. It $t$ does not scale with $K$, we have simply used
\begin{equation}
\label{eq:order_approx_noscale}
    \binom{K}{t} = \frac{K!}{t!(K-t)!} = O(K^t) \; .
\end{equation}
However, if $t$ scales with $K$ (i.e., if $\gamma = \frac{M}{N}$ does not scale with $K$), this order approximation is no longer valid. In this case, to approximate $\binom{K}{t}$, we can consider another problem where we repeat a binary experiment with the success probability of $\gamma$ for $K$ times. As $K$ grows large, according to the law of large numbers, the number of \textit{typical} sequences (i.e., sequences with $K \gamma = t$ success outcomes) would be $\binom{K}{t}$. However, from information theory, we also know that the number of typical sequences approaches $H(\gamma)$, where 
\begin{equation*}
    H(\gamma) = -\gamma \log_2 \gamma -(1-\gamma) \log_2 (1-\gamma) 
\end{equation*}
is the entropy function. Hence, when $t$ scales with $K$ we can use the approximation
%\begin{equation}
    $\binom{K}{t} = O(2^{KH(\gamma)})$.
%\end{equation}
From Table~\ref{tab:comp_scal_noscal}, it is clear that the complexity order of the proposed cyclic caching scheme is considerably smaller than all other schemes  (linear if $t$ does not scale with $K$, and quadratic otherwise).

Finally, in Table~\ref{tab:subpack_comp}, we have respectively compared the subpacketization requirement and the transmission count in some different network setups for M-S, LIN, RED, L-E, and M-B schemes.
In the table, many entries are left empty for L-E and M-B schemes due to their tight restrictions on the network parameters. The L-E scheme is originally designed for networks in which $t \ge \alpha$ and $\alpha$ divides both $t$ and $K$, while M-B requires $\frac{t+\alpha}{t+1}$ to be an integer.

From Table~\ref{tab:subpack_comp}, it is clear that except for the L-E scheme, the RED scheme has the lowest subpacketization requirement among all the schemes. 
%However, the L-E scheme has no DoF loss only when $t \ge \alpha$ and $\alpha$ divides both $t$ and $K$. This makes the \mbox{L-E} scheme less applicable for most networks with large antenna arrays, in which $\alpha \ge t$. 
Also, regarding the number of transmissions, except for a specific case (the third row in the table), LIN and RED outperform all other schemes. It can be seen that, even when the user count is as low as $K=30$, the M-S scheme becomes infeasible due to the substantial values for both the subpacketization and the transmission count. On the other hand, it is possible to implement RED even for a very large network of $K=400$ users in which the spatial DoF is $\alpha=100$. The results also show how proper tuning of $\alpha$ and $K_f$ can help further reduce both the complexity and subpacketization in the RED scheme.

%\textcolor{blue}{this part still needs updates} Clearly, except for the transmission count at specific small networks (e.g. $K=8$, $t=2$, $\alpha = 5$), LIN and RED schemes outperform M-S with a large margin. . It is also cleat how the RED scheme can benefit from reduced complexity through controlling $\alpha$ and $K_f$ parameters. %\textcolor{blue}{Pooya: Here we also need a column for the JSAC paper [10] subpacketization.}

\subsection{Simulation Results}\label{sec:sim_results}

We use numerical simulations to compare the performance of cyclic caching with other schemes. The symmetric rate, as defined in~\eqref{eq:sym_rate}, is used as the comparison metric. 
%We simulate the performance for M-S, LIN, and RED schemes. 
The L-E and M-B schemes are ignored in the simulations as they require tight restrictions on network parameters to work without significant performance (DoF) loss. Moreover, for the sake of comparison, we also consider a baseline scheme without coded caching, denoted by No-CC, in which only the local caching gain at each user is attained together with the spatial multiplexing gain. In the baseline scheme, we create $K$ transmission vectors, where users $\{1,2,...,\alpha\}$ are served during the first transmission ($i=1$), while for $i > 1$, the served user indices are a circular shift of the user indices targeted at transmission $i-1$. 
%The subpacketization requirement for the No-CC scheme is $K$, and it is \textit{not} necessarily optimal among the set of schemes achieving only the local caching gain. 
For all the simulations, we use max-min-SINR optimal (MMSE-type) beamformers, found through solving~\eqref{eq:optimizedBeamOP}. For LIN, RED, and No-CC schemes, the beamformers are designed with the uplink-downlink duality solution described in Section~\ref{section:lin_beamformer}. For the M-S scheme, we use the more complex SCA method detailed in~\cite{tolli2017multi}.

\input{Figures/SmallNet-Alpha-2-3}

As discussed earlier, the complexity of the M-S scheme makes its implementation infeasible even for moderate-sized networks. In order to be able to compare all the schemes, we consider a small network of $K=6$ users. The performance comparison results are provided in Figure~\ref{fig:smallnet_alpha-2-3}. It can be seen that the performance values of LIN and RED schemes lie between M-S and No-CC. M-S provides better performance because it benefits from a \textit{multicasting} gain, i.e., a single codeword (created with the XOR operation) benefits all the users in a multicast group. This multicasting effect is explained in more detail in~\cite{salehi2019subpacketization}, with the help of the so-called efficiency index parameter. On the other hand, the No-CC scheme has the worst performance as it lacks the coded caching gain entirely. It should also be noted that choosing a smaller $\alpha$ value  improves the performance of both \mbox{M-S} and LIN schemes at the lower SNR range, while this effect is more prominent for the LIN scheme.\footnote{Note that when $\alpha=1$, placement and delivery algorithms in M-S scheme are the same as the original scheme of~\cite{maddah2014fundamental}. The beamformer design also concedes with the max-min fair beamforming in~\cite{shariatpanahi2018physical}.} This result is in line with the findings of~\cite{tolli2017multi,lampiris2019bridging}, where a smaller $\alpha$ is shown to improve the performance at the finite-SNR regime due to a better beamforming gain. Finally, Figure~\ref{fig:smallnet_alpha-2-3} demonstrates that the RED scheme provides the same performance as the LIN scheme but with lower complexity. This near-identical performance is a result of the two schemes having a very similar coding and interference-cancellation structure.

Unfortunately, although the M-S scheme has superior performance, its high complexity makes the implementation infeasible even for slightly larger networks. For example, for the small network of 6 users considered in Figure~\ref{fig:smallnet_alpha-2-3}, with our simulation setup, the required time for simulating the M-S scheme is $O(10^3)$ times larger than the LIN scheme. As the network size grows, this ratio between the simulation times also grows exponentially. On the other hand, the simulation time for the RED scheme is roughly four times smaller than LIN. This is in line with the fact that for the considered network with $K=6,t=2,\alpha=2$, subpacketization and transmission count are reduced by a factor of $gcd(K,t,\alpha)^2=4$. %Also, as noted beforehand, the two schemes maintain an identical performance.

\input{Figures/MergeFig-KLChange}

In Figure~\ref{fig:merge_klchange}, we have analyzed the performance of the RED scheme with respect to $K$ and $L$ parameters while assuming $t=2$ and $\alpha=L$. From Figure~\ref{fig:k_change}, we see that with the DoF value $t+L$ fixed, the performance slightly reduces as $K$ increases. This is because we have assumed a fixed $t = \frac{KM}{N}$ value, which means the local caching gain $\frac{M}{N}$ is reduced as $K$ is increased. On the other hand, from Figure~\ref{fig:l_change}, we see that increasing $L$ results in a linear increase in rate. This is simply because increasing $L$ improves DoF, and accordingly, the system performance.

Figure~\ref{fig:merge_largenet} illustrates
%Figures~\ref{fig:largenet_L25} and~\ref{fig:largenet_L25_ratio} illustrate
the performance of the RED scheme for a large network with $K=100$,$t=10$, and $L=25$. To the best of our knowledge, this is the first time a coded caching scheme is applied with optimized beamformers for such a large network. In Figure~\ref{fig:largenet_L25}, the results for the No-CC scheme are included for comparison. It can be verified that decreasing $\alpha$ from 20 to 10 gives a small performance boost at the low- to moderate-SNR regime ($<15$dB) due to improved beamforming gain. On the other hand, at larger SNR values, the $\alpha=20$ setup performs much better due to the increased spatial multiplexing gain. These are in line with the findings in~\cite{tolli2017multi}.

%\begin{minipage}{\textwidth}
%    \centering
    
\input{Figures/MergeFig-LargeNet}
\input{Figures/MergeFig-Phantom}
%\end{minipage}

%\input{Figures/KF-Effect}

%\input{Figures/KF-Effect-Ratio}

Finally, the complexity reduction effect of the $K_f$ parameter is analyzed in Figure~\ref{fig:merge_phantom} %Figures~\ref{fig:kf_effect} and~\ref{fig:kf_effect_ratio} 
for two network scenarios of size $K_1=100$ and $K_2=30$ users. In both networks, we have assumed $t=7$, $L=20$ and $\alpha = 14$. Without any phantom users, $gcd(K,t,\alpha)=1$ and the subpacketization values for the two networks are $S_1 = 2100$ and $S_2=630$, respectively. However, by adding only $K_f=5$ phantom users, the common denominator becomes $gcd(K+K_f,t,\alpha)=7$, and hence, the subpacketization values are reduced to $S_1 = 45$ and $S_2 = 15$, respectively. Interestingly, the decrease in the achievable rate due to adding $K_f$ phantom users is relatively minor for both networks. In fact, for the larger network, the performance loss is less than 4\% over the entire SNR range, while for the smaller network, the deterioration is less than 15\%. On the other hand, adding the extra phantom users reduces the simulation time in our setup by a factor of $\sim 10$. %This clarifies how cyclic caching enables controlling the complexity by joint adjustment of $K_f$ and $\alpha$ parameters. 

%% file: Figures/SmallNet-Alpha-2-3.tex
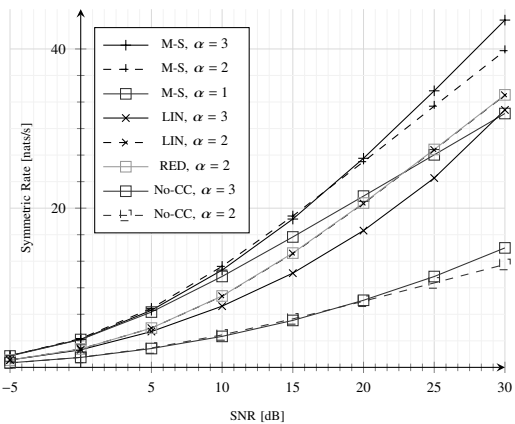
\begin{figure}
    \centering
    %\resizebox{0.5\columnwidth}{!}{%

    \begin{tikzpicture}

    \begin{axis}
    [
    width = 0.45\textwidth,
    height = 0.35\textwidth,
    % put axis lines at left and bottom
    axis lines = center,
    % control axis labels
    xlabel near ticks,
    xlabel = \smaller {\tiny SNR [dB]},
    ylabel = \smaller {\tiny Symmetric Rate [nats/s]},
    ylabel near ticks,
    ymin = 0, ymax = 45,
    % control legend position
    legend style={at={(0.17,0.95)},anchor=north west},
    %legend pos= north west,
    % control size of tick marks (10,20,30,etc)
    ticklabel style={font=\tiny},
    % control major grids
    grid=both,
    major grid style={line width=.2pt,draw=gray!30},
    % control minor grids
    grid style={line width=.1pt, draw=gray!10},
    minor tick num=5,
    ]
    
    %\addplot
    %[mark = none, dashed, black!40]
    %table[y=Bound-A3,x=SNR]{Data/SN1-3-Rate.tex};
    %\addlegendentry{\tiny Bound, $\alpha=3$}
    
    \addplot
    [mark = +, black]
    table[y=MS-A3,x=SNR]{Data/SN1-3-Rate.tex};
    \addlegendentry{\tiny M-S, $\alpha=3$}
    
    \addplot
    [mark = +, dashed, black]
    table[y=MS-A2,x=SNR]{Data/SN1-3-Rate.tex};
    \addlegendentry{\tiny M-S, $\alpha=2$}
    
    \addplot
    [mark = square, black!80]
    table[y=MS-A1,x=SNR]{Data/SN1-3-Rate.tex};
    \addlegendentry{\tiny M-S, $\alpha=1$}
    
    \addplot
    [mark = x, black]
    table[y=LIN-A3,x=SNR]{Data/SN1-3-Rate.tex};
    \addlegendentry{\tiny LIN, $\alpha = 3$}
    
    \addplot
    [mark = x, dashed, black]
    table[y=LIN-A2,x=SNR]{Data/SN1-3-Rate.tex};
    \addlegendentry{\tiny LIN, $\alpha = 2$}
    
    \addplot
    [mark = square, black!40]
    table[y=RED-A2,x=SNR]{Data/SN1-3-Rate.tex};
    \addlegendentry{\tiny RED, $\alpha = 2$}
    
    \addplot
    [mark = square, black!80]
    table[y=NOC-A3,x=SNR]{Data/SN1-3-Rate.tex};
    \addlegendentry{\tiny No-CC, $\alpha=3$}
    
    \addplot
    [mark = square, dashed, black!80]
    table[y=NOC-A2,x=SNR]{Data/SN1-3-Rate.tex};
    \addlegendentry{\tiny No-CC, $\alpha=2$}
    
    \end{axis}

    \end{tikzpicture}
    %}

    \caption{M-S vs LIN vs RED vs No-CC rate; $K=6$, $t=2$, $L=3$}
    \label{fig:smallnet_alpha-2-3}
\end{figure}

%% file: Figures/MergeFig-KLChange.tex
\begin{figure*}
\centering
\begin{subfigure}{.5\textwidth}

\input{Figures/K-Change}

\end{subfigure}%
\begin{subfigure}{.5\textwidth}

\input{Figures/L-Change}

\end{subfigure}
\caption{Performance of RED vs $K$ and $L$ parameters, $t=2$, $\alpha=L$, SNR$=10\mathrm{dB}$}
\label{fig:merge_klchange}
\end{figure*}

%% file: Figures/K-Change.tex
%\begin{figure}
    \centering
    %\resizebox{0.85\columnwidth}{!}{%

    \begin{tikzpicture}

    \begin{axis}
    [
    width = 0.90\columnwidth,
    height = 0.60\columnwidth,
    % put axis lines at left and bottom
    axis lines = center,
    % control axis labels
    xlabel near ticks,
    xlabel = \smaller {\tiny User Count ($K$)},
    ylabel = \smaller {\tiny Symmetric Rate [nats/s]},
    ylabel near ticks,
    ymin = 0, ymax = 10, xmax=30,
    % control legend position
    %legend pos = north west,
    legend style={at={(0.2,0.95)},anchor=north west},
    % control size of tick marks (10,20,30,etc)
    ticklabel style={font=\tiny},
    % control major grids
    grid=both,
    major grid style={line width=.2pt,draw=gray!30},
    % control minor grids
    %grid style={line width=.1pt, draw=gray!10},
    %minor tick num=5,
    ]
    
    \addplot
    [mark = x, black]
    table[y=TotalRate,x=K]{Data/K-Change.tex};
    %\addlegendentry{\tiny $K=100$, $K_f=0$}
    
    %\addplot
    %[black!50]
    %table[y=RED-KF5,x=SNR]{Data/SN7-8-K100-Rate.tex};
    %\addlegendentry{\tiny $K=100$, $K_f=5$}
    
    \end{axis}

    \end{tikzpicture}
    %}

    \caption{Performance vs $K$, \; $L = 4$
    }
    \label{fig:k_change}
%\end{figure}

%% file: Figures/L-Change.tex
%\begin{figure}
    \centering
    %\resizebox{0.85\columnwidth}{!}{%

    \begin{tikzpicture}

    \begin{axis}
    [
    width = 0.90\columnwidth,
    height = 0.60\columnwidth,
    % put axis lines at left and bottom
    axis lines = center,
    % control axis labels
    xlabel near ticks,
    xlabel = \smaller {\tiny Antenna Count ($L$)},
    ylabel = \smaller {\tiny Symmetric Rate [nats/s]},
    ylabel near ticks,
    ymin = 0, xmax=24,
    % control legend position
    %legend pos = north west,
    legend style={at={(0.2,0.95)},anchor=north west},
    % control size of tick marks (10,20,30,etc)
    ticklabel style={font=\tiny},
    % control major grids
    grid=both,
    major grid style={line width=.2pt,draw=gray!30},
    % control minor grids
    %grid style={line width=.1pt, draw=gray!10},
    %minor tick num=5,
    ]
    
    \addplot
    [mark = x, black]
    table[y=TotalRate,x=L]{Data/L-Change.tex};
    %\addlegendentry{\tiny $K=100$, $K_f=0$}
    
    %\addplot
    %[black!50]
    %table[y=RED-KF5,x=SNR]{Data/SN7-8-K100-Rate.tex};
    %\addlegendentry{\tiny $K=100$, $K_f=5$}
    
    \end{axis}

    \end{tikzpicture}
    %}

    \caption{Performance vs $L$, \; $K = 30$
    }
    \label{fig:l_change}
%\end{figure}

%% file: Figures/MergeFig-LargeNet.tex
\begin{figure*}
\centering
\begin{subfigure}{.5\textwidth}

\input{Figures/LargeNet-L25}

\end{subfigure}%
\begin{subfigure}{.5\textwidth}

\input{Figures/LargeNet-L25-Ratio}

\end{subfigure}
\caption{Performance of RED for a large network, $K=100$, $t=10$, $L=25$}
%\vspace{-10pt}
\label{fig:merge_largenet}
\end{figure*}

%% file: Figures/LargeNet-L25.tex
%\begin{figure}
    \centering
    %\resizebox{0.85\columnwidth}{!}{%

    \begin{tikzpicture}

    \begin{axis}
    [
    width = 0.90\columnwidth,
    height = 0.70\columnwidth,
    % put axis lines at left and bottom
    axis lines = center,
    % control axis labels
    xlabel near ticks,
    xlabel = \smaller {\tiny SNR [dB]},
    ylabel = \smaller {\tiny Symmetric Rate [nats/s]},
    ylabel near ticks,
    ymin = 0, ymax = 200, xmax=30,
    % control legend position
    %legend pos = north west,
    legend style={at={(0.2,0.95)},anchor=north west},
    % control size of tick marks (10,20,30,etc)
    ticklabel style={font=\tiny},
    % control major grids
    grid=both,
    major grid style={line width=.2pt,draw=gray!30},
    % control minor grids
    %grid style={line width=.1pt, draw=gray!10},
    %minor tick num=5,
    ]
    
    \addplot
    [mark = x, dashed, black]
    table[y=RED-A10,x=SNR]{Data/SN4-6-L25-Rate.tex};
    \addlegendentry{\tiny RED, $\alpha=10$}
    
    \addplot
    [black!50]
    table[y=RED-A15,x=SNR]{Data/SN4-6-L25-Rate.tex};
    \addlegendentry{\tiny RED, $\alpha=15$}
    
    \addplot
    [mark = +, black!80]
    table[y=RED-A20,x=SNR]{Data/SN4-6-L25-Rate.tex};
    \addlegendentry{\tiny RED, $\alpha=20$}
    
    \addplot
    [mark = square, dashed, black]
    table[y=NOC-A10,x=SNR]{Data/SN4-6-L25-Rate.tex};
    \addlegendentry{\tiny No-CC, $\alpha=10$}
    
    \addplot
    [dashed, black!50]
    table[y=NOC-A15,x=SNR]{Data/SN4-6-L25-Rate.tex};
    \addlegendentry{\tiny No-CC, $\alpha=15$}
    
    \addplot
    [mark = x, black!80]
    table[y=NOC-A20,x=SNR]{Data/SN4-6-L25-Rate.tex};
    \addlegendentry{\tiny No-CC, $\alpha=20$}
    
    \end{axis}

    \end{tikzpicture}
    %}

    \caption{Performance vs $\alpha$%; $K=100$, $t=10$, $L=25$
    }
    \label{fig:largenet_L25}
%\end{figure}

%% file: Figures/LargeNet-L25-Ratio.tex
%\begin{figure}
    \centering
    %\resizebox{0.91\columnwidth}{!}{%

    \begin{tikzpicture}

    \begin{axis}
    [
    width = 0.90\columnwidth,
    height = 0.70\columnwidth,
    % put axis lines at left and bottom
    axis lines = left,
    % control axis labels
    xlabel near ticks,
    xlabel = \smaller {\tiny SNR [dB]},
    ylabel = \smaller {\tiny Performance Ratio [\%]},
    ylabel near ticks,
    ymin = -5, ymax = 20,
    % control legend position
    legend pos = north west,
    % control size of tick marks (10,20,30,etc)
    ticklabel style={font=\tiny},
    % control major grids
    grid=both,
    major grid style={line width=.2pt,draw=gray!30},
    % control minor grids
    %grid style={line width=.1pt, draw=gray!10},
    %minor tick num=5,
    ybar interval=0.7,
    ]
    
    \addplot
    [black, fill=black!40]
    table[y=RED15-10,x=SNR]{Data/SN4-6-L25-Rate.tex};
    \addlegendentry{\tiny $\alpha=15$}
    
    \addplot
    [black, fill=black!20]
    table[y=RED20-10,x=SNR]{Data/SN4-6-L25-Rate.tex};
    \addlegendentry{\tiny $\alpha=20$}
    
    \end{axis}

    \end{tikzpicture}
    %}

    \caption{Performance ratio w.r.t. $\alpha=10$%, RED; $K=100$, $t=10$, $L=25$
    }
    \label{fig:largenet_L25_ratio}
%\end{figure}

%% file: Figures/MergeFig-Phantom.tex
\begin{figure*}
\centering
\begin{subfigure}{.5\textwidth}

\input{Figures/KF-Effect}

\end{subfigure}%
\begin{subfigure}{.5\textwidth}

\input{Figures/KF-Effect-Ratio}

\end{subfigure}
\caption{The effect of the $K_f$ parameter, RED scheme, $t=7$, $L=20$, $\alpha = 14$}
\label{fig:merge_phantom}
\end{figure*}

%% file: Figures/KF-Effect.tex
%\begin{figure}
    \centering
    %\resizebox{0.85\columnwidth}{!}{%

    \begin{tikzpicture}

    \begin{axis}
    [
    width = 0.90\columnwidth,
    height = 0.70\columnwidth,
    % put axis lines at left and bottom
    axis lines = center,
    % control axis labels
    xlabel near ticks,
    xlabel = \tiny {SNR [dB]},
    ylabel = \tiny {Symmetric Rate [nats/s]},
    ylabel near ticks,
    ymin = 0, ymax = 160, xmax=30,
    % control legend position
    %legend pos = north west,
    legend style={at={(0.2,0.95)},anchor=north west},
    % control size of tick marks (10,20,30,etc)
    ticklabel style={font=\tiny},
    % control major grids
    grid=both,
    major grid style={line width=.2pt,draw=gray!30},
    % control minor grids
    %grid style={line width=.1pt, draw=gray!10},
    %minor tick num=5,
    ]
    
    \addplot
    [mark = x, dashed, black]
    table[y=RED-KF0,x=SNR]{Data/SN7-8-K100-Rate.tex};
    \addlegendentry{\tiny $K=100$, $K_f=0$}
    
    \addplot
    [black!50]
    table[y=RED-KF5,x=SNR]{Data/SN7-8-K100-Rate.tex};
    \addlegendentry{\tiny $K=100$, $K_f=5$}
    
    \addplot
    [mark = +, black!80]
    table[y=RED-KF0,x=SNR]{Data/SN7-8-K30-Rate.tex};
    \addlegendentry{\tiny $K=30$, $K_f=0$}
    
    \addplot
    [mark = square, dashed, black]
    table[y=RED-KF5,x=SNR]{Data/SN7-8-K30-Rate.tex};
    \addlegendentry{\tiny $K=30$, $K_f=5$}
    
    \end{axis}

    \end{tikzpicture}
    %}

    \caption{Performance vs $K_f$%, RED; $t=7$, $L=20$, $\alpha=14$
    }
    \label{fig:kf_effect}
%\end{figure}

%% file: Figures/KF-Effect-Ratio.tex
%\begin{figure}
    \centering
    %\resizebox{0.90\columnwidth}{!}{%

    \begin{tikzpicture}

    \begin{axis}
    [
    width = 0.90\columnwidth,
    height = 0.70\columnwidth,
    % put axis lines at left and bottom
    axis lines = left,
    % control axis labels
    xlabel near ticks,
    xlabel = \tiny {SNR [dB]},
    ylabel = \tiny {Performance Ratio [\%]},
    ylabel near ticks,
    ymin = -15, ymax = 10,
    % control legend position
    legend pos = north west,
    % control size of tick marks (10,20,30,etc)
    ticklabel style={font=\tiny},
    % control major grids
    grid=both,
    major grid style={line width=.2pt,draw=gray!30},
    % control minor grids
    %grid style={line width=.1pt, draw=gray!10},
    %minor tick num=5,
    ybar interval=0.7,
    ]
    
    \addplot
    [black, fill=black!40]
    table[y=RED5-0,x=SNR]{Data/SN7-8-K100-Rate.tex};
    \addlegendentry{\tiny $K=100$, $K_f=5$}
    
    \addplot
    [black, fill=black!20]
    table[y=RED5-0,x=SNR]{Data/SN7-8-K30-Rate.tex};
    \addlegendentry{\tiny $K=30$, $K_f=5$}
    
    \end{axis}

    \end{tikzpicture}
    %}

    \caption{Performance ratio w.r.t. $K_f=0$%, RED; $t=7$, $L=20$, $\alpha=14$
    }
    \label{fig:kf_effect_ratio}
%\end{figure}

%% file: LinearS-ValidityComplexity.tex
%\subsection{Validity}
\label{section:ICCvalidity}
For the delivery and decoding procedures presented in Sections~\ref{section:lin_delivery} and~\ref{section:lin_decoding}, we here show that at the end of the $K$ rounds, each user successfully decodes all its missing data parts. 
%Following the description of the achievable scheme, in particular the codewords generation and the decoding procedure, we can immediately conclude that each transmission serves $t+L$ subfiles to $t+L$ different users. 
We recall that for the placement matrix $\BV$, if $\BV[p,k] = 0$ for some $p \in [K]$ and $k \in [K]$, then the packet $W_p(k)$, which is comprised of $t+\alpha$ smaller subpackets of the form $W_p^q(k)$, must be delivered to user $k$. 
%each zero entry in the placement matrix $\BV$ represents $t+\alpha$ data parts, where each data part is denoted as $W_p^q$, of the total number of subfiles that must be successfully delivered to the corresponding demanding users. 
By construction of the user index and packet index vectors $\Bk_j^r$ and $\Bp_j^r$, where $r\in[K]$ and $j\in [t+\alpha]$, it is easy to see that for any $p\in[K]$ and $k\in[K]$, if $\BV[p,k]=0$ there exists always a transmission vector that contains one of the $t+\alpha$ subpackets represented by $\BV[p,k]$; i.e. there exists always a triple $(j,r,n)$ such that $p=\Bp_j^r[n]$ and $k=\Bk_j^r[n]$.
Therefore, it is sufficient to show that if $\BV[p,k]=0$ for some pair $(p,k)$, there exist exactly $t+\alpha$ triples $(j,r,n)$ for which $p=\Bp_j^r[n]$ and $k=\Bk_j^r[n]$.

Without loss of generality, let us consider the first round of the delivery scheme. The vector $\Bk_j^1$ contains the index of the users targeted at transmission $j$ of this round. Let us split the user indices in $\Bk_j^1$ into the following two vectors
%and let us rewrite the set $\Bk_j^1$ of users targeted in transmission $j$ of round $1$  as the union of the 2 following (ordered) sets:
\begin{equation*}
    \begin{aligned}
&\Bc_j^1=[1:t] \; , \quad
\Bs_j^1=[(([1:\alpha]+j-1)\%(K-t))+t] \; .
\end{aligned}
\end{equation*}
Similarly, we split the packet indices in $\Bp_j^1$ into the following two vectors
%we can rewrite the packet index vector of round $1$ $\Bp_j^1$ as
\begin{align}
&\Br_j^1=[((t+j-[1:t])\%(K-t))+[1:t]] \; ,  \quad
\Bd_j^1=\Be(\alpha) \; .   \nonumber
\end{align}
Now let us consider the set of points in the tabular representation corresponding to the vector pair $(\Br_1^j,\Bc_1^j)$; i.e., the set of points in which the row index is taken from $\Br_1^j$ and column index is taken from $\Bc_1^j$. We will use $\{(\Br_1^j,\Bc_1^j)\}$ to denote this set. From the graphical example in Section~\ref{section:graphical_rep}, we recall that $\{(\Br_1^j,\Bc_1^j)\}$ is an element-wise circular shift of $\{(\Br_1^1,\Bc_1^1)\}$, over the non-shaded cells of the tabular representation and in the vertical direction.
%; i.e., $\CV_1^1$ and $\CL_1^j$ are obtained by adding one unit to each element of\footnote{This explanation is simplified because the light-shaded entries of the tabular representation are actually not allowed (because they represent subpackets already stored in the cache of the corresponding users). } $\CL_1^1$. 
Similarly, $\{(\Bd_1^j,\Bs_1^j)\}$ is an element-wise circular shift of $\{(\Bd_1^1,\Bs_1^1)\}$, over the non-shaded cells and in the horizontal direction. As a result, in round $1$, the following two statements hold:
\begin{itemize}
    \item for every user $k\! \in\! [K]$ and packet index $p\! \in\! [K]$, if $k$ is in $\Bc_1^1=\Bc_2^1=...=\Bc_{K-t}^1$ and $p$ is in $[\Br_1^1;...;\Br_{K-t}^1]$, there exists exactly one transmission index $j \in [K-t]$, such that the $j$-th transmission of the first round delivers $W_p^q(k)$ to user $k$, for some subpacket index $q \in [t+\alpha]$. In other words there exists exactly one triple $(j,1,n)$ for the pair $(p,k)$;
    \item for every user $k \in [K]$, if $k$ is in $[\Bs_1^1;\Bs_2^1;...;\Bs_{K-t}^1]$, there exist exactly $\alpha$ transmissions in the first round that deliver $W_1^q(k)$ to user $k$, each with a distinct subpacket index $q \in [t+\alpha]$. In other words, there exist exactly $\alpha$ triples $(j,1,n)$, for the pair $(1,k)$.
\end{itemize}
These statements also hold for other transmission rounds, i.e., transmission round $r$ where $1 < r \le K$. However, for any $k,p \in [K]$ such that $\BV [p,k] = 0$, there exists exactly one round $r$ for which $k$ is in $[\Bs_1^r;\Bs_2^r;...;\Bs_{K-t}^r]$, while there exist $t$ different rounds $r$ for which $k$ is in $\Bc_1^r=\Bc_2^r=...=\Bc_{K-t}^r$. Hence, for any $k,p \in [K]$ such that $\BV [p,k] = 0$, there exist 
\begin{equation*}
    \alpha \times 1 + 1 \times t = \alpha+t
\end{equation*}
triples $(j,r,n)$ for the pair $(p,k)$. This clarifies the correctness of the delivery algorithm proposed in Section~\ref{section:lin_delivery}.

%Finally, for any $k,p \in [K]$ such that $\BV [p,k] = 0$, there exists exactly one round $r$ for which $k$ appears in $\bigcup_{j=1}^{K-t} \CH_r^j$; but there exists $t$ different rounds $r$ for which $k$ appears in $\bigcup_{j=1}^{K-t}\CV_r^j=\CV_r^1$. Thus, for a fixed $p\in[K]$ and a fixed $k\in[K]$ there exist $\alpha \times 1 + 1 \times t = \alpha+t$ transmissions 
%\footnote{We recall that, in each transmission that serves a subfile with lower index $p$ intended for user $k$, the choice of the upper index $q$, that defines which of the $t+\alpha$ subfiles $W_p^q(k)$ of $W_p(k)$ is transmitted, is sequential.} 
%in which all $t+\alpha$ subpackets $W^q_p(k)$ are transmitted to user $k$.

%% file: FurtherReduction.tex
\label{secion:RED_details}
To reduce the subpacketization requirement, we apply a user grouping mechanism, inspired by~\cite{lampiris2018adding}. For notation simplicity, let us use simply use $\phi$ to represent $\phi_{K,t,\alpha}$. The idea is to split users into groups of size $ \phi_{K,t,\alpha}$ and assume each group is equivalent to a \textit{virtual} user. Then, we consider a virtual network consisting of these virtual users, in which the coded caching and spatial multiplexing gains are $\frac{t}{\phi}$ and $\frac{\alpha}{\phi}$, respectively. Finally, we apply the coded caching scheme proposed in Sections~\ref{section:lin_placement} and~\ref{section:lin_delivery} to the virtual network, and \textit{elevate} the resulting cache placement and delivery schemes to be applicable in the original network. Here, we provide a detailed description of the elevation procedure.

\subsubsection{Cache Placement}
%To reduce the subpacketization, we follow a user grouping technique inspired by~\cite{lampiris2018adding}. 
Assume the original network is given with $K$ users, coded caching gain of $t$, and spatial DoF of $\alpha$. 
We first split the set of users $[K]$ into $K'=\frac{K}{\phi}$ disjoints groups $v_{k'}$, $k'\in[K']$, where all groups have the same number of $\phi$ users. Without loss of generality, we assume each group $v_{k'}$ corresponds to the set of users  
\begin{equation}
\label{eq:virtual_original_user_relation}
    v_{k'} \triangleq [\phi*(k'-1)+1:\phi*k']\;, \; \; \forall k' \in [K'] \; .
\end{equation}
Next, we assume each user group is equivalent to a virtual user with cache size of $Mf$ bits, and the virtual users form a new virtual network in which the spatial DoF is $\alpha' = \frac{\alpha}{\phi}$. For this virtual network, we use the same cache placement algorithm presented in Section~\ref{section:lin_placement}, and set the cache content of every user in group $v_{k'}$ to be the same as the cache content of the virtual user corresponding to $v_{k'}$. The total subpacketization is then $K'(t'+\alpha')$, where $t' = \frac{t}{\phi}$.

%In the cache placement phase we design the placement matrix as in Section \ref{section:ICCPaper} for the virtual network with $K'$ users and assuming that the normalized sum cache size is $t'=\frac{t}{\phi}$. Each real user $k\in v_{k'}$ will store in its cache the same content that would be cached by virtual user $v_{k'}$ according to the designed placement matrix. Then, each of the $K'$ subfiles of each file is further split into $t'+\alpha'$ smaller parts.

\subsubsection{Delivery Phase}
During the delivery phase, we first create transmission vectors for the virtual network, and then \textit{elevate} them to be applicable in the original network. Following~\eqref{eq:general_transmission_vector_model}, the transmission vector $i$ for the virtual network is built as
%\begin{equation}
%\label{eq:general_transmission_vector_virtual}
    $\Bx_i' = \sum_{v_{k'} \in \CX_i'} \Bw_i'(v_{k'}) X_i'(v_{k'})$,
%\end{equation}
in which $\CX_i'$ is the set of virtual users targeted at transmission $i$, $X_i'(v_{k'})$ denotes the data part targeted to the virtual user $v_{k'}$ during the same transmission, and $\Bw_i'(v_{k'})$ is the beamformer vector assigned to $X_i'(v_{k'})$. In order to elevate~$\Bx_i'$ for the original network, we first notice than each virtual user represents a set of $\phi$ original users, and hence, using~\eqref{eq:virtual_original_user_relation}, the set of targeted users during transmission $i$ for the original network is
\begin{equation}
\label{eq:target_user_set_virtual_to_original}
    \CX_i = \bigcup\limits_{v_{k'} \in \CX_i'}  [\phi*(k'-1)+1:\phi*k'] \; .
\end{equation}
Following the discussions in Section~\ref{section:lin_beamformer}, every beamformer vector $\Bw_i'(v_{k'})$ is built to suppress unwanted terms at $\alpha'-1$ virtual users. Let us denote the set of such virtual users as $\CR_i'(v_{k'})$, where $|\CR_i'(v_{k'})| = \alpha'-1$. The goal is then to find the respective set for the original network, denoted by $\CR_i(k)$, for $k \in  [K]$. As the spatial DoF for the original network is $\alpha$, it is possible to suppress undesired terms at $\alpha -1$ original users; i.e. $|\CR_i(k)| = \alpha - 1$. Without loss of generality, let us assume that user $k$ is in the respective group of $v_{k'}$ (every user in the original network has one counterpart in the virtual network). In the original network, for the interference to be suppressed, the following conditions should be met:
\begin{enumerate}
    \item For every ${v}_{\hat{k'}} \in \CR_i'(v_{k'})$, $\CR_i(k)$ should include all the users in the respective group of ${v}_{\hat{k'}}$; i.e. $\CR_i(k)$ should include all the users in $[\phi*(\hat{k'}-1)+1:\phi*\hat{k'}]$;
    \item $\CR_i(k)$ should include all other users in the respective group of $v_{k'}$; i.e. $\CR_i(k)$ should include all the users in $[\phi*(k'-1)+1:\phi*k'] \backslash \{k\}$.
\end{enumerate}
Using a formal representation, we have
%Clearly, for every ${v}_{\hat{k'}} \in \CR_i'(v_{k'})$, $\CR_i(k)$ should include all the users in $[\phi*(\hat{k'}-1)+1:\phi*\hat{k'}]$. However, the intended codeword for user $k$, denoted by $X_i(k)$, should not only be zero-forced (suppressed) at all such users, but should also be zero-forced (suppressed) at all users in $[\phi*(k'-1)+1:\phi*k'] \backslash \{k\}$, for the interference to be completely eliminated. This means for $k \in [\phi*(k'-1)+1:\phi*k']$, $\CR_i(k)$ is built as
\begin{equation}
\label{eq:zero_foce_set_virtual_to_original}
\begin{aligned}
    \CR_i(k) = &[\phi*(k'-1)+1:\phi*k'] \backslash \{ k \} \\
    &\bigcup\limits_{v_{\hat{k'}} \in \CR_i'(v_{k'})}  [\phi*(\hat{k'}-1)+1:\phi*\hat{k'}] \; .
    %\bigcup \quad & \; .
\end{aligned}
\end{equation}
In other words, the data part $X_i(k)$ intended for user $k$ at transmission $i$ has to be suppressed not only at $\alpha'-1$ virtual users (where each virtual user represents $\phi$ original users), but also at $\phi-1$ original users in the equivalent group of $v_{k'}$. Hence, the total number of users for which $X_i(k)$ should be suppressed is $(\alpha'-1)\phi + \phi-1$. Substituting $\alpha' = \frac{\alpha}{\phi}$, we have
\begin{equation}
    |\CR_i(k)| = (\frac{\alpha}{\phi}-1)\phi + \phi-1 = \alpha-1 \; .
\end{equation}

Now, we can elevate $\Bx_i'$ to be applicable in the original network. All we need to do is to  use~\eqref{eq:target_user_set_virtual_to_original} to substitute the target user set $\CX_i'$ with $\CX_i$, and replace $\Bw_i'(v_{k'}) X_i'(v_{k'})$ with
%Summarizing the discussions, in order to modify the transmission vector model for the virtual network, as presented in \eqref{eq:general_transmission_vector_virtual}, to fit the original model of \eqref{eq:general_transmission_vector_model}, one needs to substitute the target user set $\CX_i'$ with $\CX_i$ based on \eqref{eq:target_user_set_virtual_to_original}; and $\Bw_i'(v_{k'}) X_i'(v_{k'})$ with
\begin{equation}
    \sum_{k \in [\phi*(k'-1)+1:\phi*k']} \Bw_i(k) X_i(k) \; ,
\end{equation}
where $\Bw_i(k)$ is the beamformer vector designed to suppress unwanted data terms at the user set $\CR_i(k)$ as defined in~\eqref{eq:zero_foce_set_virtual_to_original}, and $X_i(k)$ is the data part intended for user $k$ at transmission $i$. The subpacketization for the virtual network, $K'(t'+\alpha')$, would then be still valid for the original network, indicating a reduction of $\phi^2$ compared with the case when no grouping is applied.